\documentclass[review]{elsarticle}

\usepackage{hyperref}

\journal{Nonlinear Analysis: Hybrid Systems}

\usepackage{amsmath}
\usepackage{latexsym, amssymb}
\usepackage{graphicx}
\usepackage{amsmath, amsbsy}
\usepackage{amsopn, amstext}
\usepackage{hyperref}
\usepackage{cancel, color}
\usepackage{epstopdf}
\usepackage{algorithm}
\usepackage{array}
\usepackage{ragged2e}
\usepackage{multirow}
\usepackage{algorithmic}
\usepackage{mathrsfs}
\usepackage{hyperref}
\usepackage{bm}
\usepackage{algorithm}
\usepackage{ifpdf,hyperref}
\usepackage[textwidth=16cm,textheight=21.5cm]{geometry}
\usepackage{tikz} 
\usetikzlibrary{arrows}
\usetikzlibrary{calc,through}

\def\cal{\mathcal}

\def\diag{diag}

\def\0{{\bf 0}}

\newcommand{\R}{{\mathbb R}}

\newtheorem{thm}{Theorem}[section]
\newtheorem{cor}[thm]{Corollary}
\newtheorem{dfn}[thm]{Definition}

\newtheorem{exa}[thm]{Example}

\newtheorem{rem}[thm]{Remark}
\newtheorem{lem}[thm]{Lemma}



\begin{document}

\begin{frontmatter}

\title{Optimal Control of Switched Systems Governed by Logical Switching Dynamics\tnoteref{mytitlenote}}
\tnotetext[mytitlenote]{This work is supported partly by the National Natural Science Foundation of China under Grant 62103305 and 62573261, the Hong Kong Polytechnic University under Grant (1-WZ0E, 4-ZZPT), the Taishan Scholar Project of Shandong Province under Grant tsqn202312033, and the Xiaomi Young Talents Program.}

%\author{ Yanan Pan\fnref{myfootnote}}
%\address{ College of Electrical Engineering and Automation, Shandong University of Science and Technology, Qingdao, 266590, Shandong, People's Republic of China}
\author[mymainaddress]{Xiao Zhang}\ead{xiaozhang@amss.ac.cn}
\author[tongji1,tongji2]{Min Meng}\ead{mengmin@tongji.edu.cn}
\author[sdu]{Changxi Li}\ead{lichangxi@sdu.edu.cn}
\author[mymainaddress]{Ka-Fai Cedric Yiu\corref{mycorrespondingauthor}}\ead{macyiu@polyu.edu.hk}

\cortext[mycorrespondingauthor]{Corresponding author}

\address[mymainaddress]{Department of Applied Mathematics, Hong Kong Polytechnic University, Hong Kong, P.R. China}
\address[tongji1]{Department of Control Science and Engineering, College of Electronics and Information Engineering, Shanghai Research Institute for Intelligent Autonomous Systems, Tongji University, Shanghai, P.R. China}
\address[tongji2]{National Key Laboratory of Autonomous Intelligent Unmanned Systems, and Frontiers Science Center for Intelligent Autonomous Systems, Ministry of Education, Shanghai, P.R. China}
\address[sdu]{chool of Mathematics, Shandong University, Jinan, P. R. China}
\begin{abstract}
This paper investigates the optimal co-design of logical and continuous controls for switched linear systems governed by controlled logical switching dynamics. Unlike traditional switched systems with arbitrary or state-dependent switching, the switching signals here are generated by an internal logical dynamical system and explicitly integrated into the control synthesis. By leveraging the semi-tensor product (STP) of matrices, we embed the coupled logical and continuous dynamics into a unified algebraic state-space representation, transforming the co-design problem into a tractable linear-quadratic framework. We derive Riccati-type backward recursions for both deterministic and stochastic logical dynamics, which yield optimal state-feedback laws for continuous control alongside value-function-based, state-dependent decision rules for logical switching. To mitigate the combinatorial explosion inherent in logical decision-making, a hierarchical algorithm is developed to decouple offline precomputation from efficient online execution. Numerical simulations demonstrate the efficacy of the proposed framework.
\end{abstract}


\end{frontmatter}

%\linenumbers
\section{Introduction}
Modern hybrid systems in cyber-physical applications, such as autonomous vehicles and smart grids, rely on the seamless integration of discrete supervisory logic and continuous regulation. Critical operational decisions, ranging from lane-change maneuvers in traffic flow \cite{lyg98,hav14,mos23} to mode-switching strategies in power distribution \cite{dou16,wan19,nog20}, must be executed in real time by accounting for the evolving physical states. Crucially, these logical decisions are often governed by underlying discrete-event structures, where the switching signal evolves according to an internal \emph{controlled logical dynamical system} rather than being arbitrarily assigned \cite{bem99,dav00,guo22}. Consequently, achieving superior system performance necessitates a unified co-design approach that optimizes logical steering inputs and continuous regulation signals simultaneously, rather than treating switching behaviors as exogenous or pre-determined.

The intrinsic difficulty of this co-design problem lies in the structural coupling between logical evolution and continuous regulation. Unlike classical switched optimal control, where the switching index is treated as a direct decision variable at each stage, here the switching behavior is generated by an internal logical dynamical system whose state evolves according to controlled transition rules. As a result, the feasible switching sequence is implicitly constrained by the trajectory of the logical state.

From an optimization perspective, this coupling prevents the Bellman recursion from being formulated solely over the original continuous state space. The logical state must be incorporated as part of the system state, thereby enlarging the optimization domain and introducing combinatorial growth in admissible logical actions. Consequently, classical optimal control techniques cannot be directly applied without explicitly embedding logical dynamics into the state-space representation. Without such a unified framework, obtaining globally optimal solutions remains both analytically and computationally challenging.

Early influential works on optimal switching treat the switching signal as a direct or relaxed decision variable. Representative examples include embedding methods that optimize mode sequences and switching times via continuous relaxations \cite{wu19}, and the discrete-time switched LQR (DSLQR) framework, which freely selects the switching index at each time step under a quadratic cost \cite{zha09}. A dominant paradigm in optimal switching, exemplified by receding-horizon control with dynamic programming \cite{gor11}, is to jointly optimize the continuous input and a sequence of switching decisions over a finite horizon, thereby treating switching as a direct—albeit combinatorially complex—optimization variable.

While effective in minimizing performance indices, these approaches share a fundamental limitation: the switching index is assumed to be directly selectable at each decision epoch. In many practical hybrid systems, however, switching behavior is generated by an internal logical layer governed by causal transition rules rather than instantaneous optimization decisions. Mode transitions may depend on prior actions, protocol constraints, safety interlocks, or dwell-time requirements, and therefore cannot be reassigned arbitrarily.

Without explicitly incorporating this logical evolution into the system state, such intrinsic constraints cannot be systematically embedded into the optimality conditions and must instead be imposed externally or approximated heuristically. In contrast to classical switched LQR formulations, where the switching index is optimized stage by stage, here the admissible switching sequence is implicitly restricted by the evolution of a controlled logical dynamical system.

Recent studies have begun to explicitly model logical dynamics within hybrid systems. For instance, \cite{gon22} embeds logical and continuous states into a unified algebraic representation, but logical objectives are typically treated as separate reachability or constraint problems, decoupled from the overall performance index. Model Predictive Control (MPC) approaches formulate hybrid co-design as online mixed-integer programs \cite{kha25}, shifting combinatorial complexity to the real-time solver and limiting scalability. In parallel, switched systems and Markov jump linear systems (MJLS) \cite{fra01,cos05,lee08,hou10,zha26} model switching as autonomous or stochastic processes; however, since the switching evolution is exogenous, only continuous control admits feedback optimization. 

Motivated by the discrepancy between unconstrained switching assumptions and practical system constraints, as well as the computational limitations of existing logical modeling techniques, this paper develops a unified optimal control framework for discrete-time switched linear systems governed by controlled logical switching dynamics. The overall system architecture, illustrating the interaction between logical and continuous layers, is depicted in Fig. \ref{fig:system_structure}. Leveraging the semi-tensor product (STP) of matrices, we establish an algebraic state-space representation (ASSR) that embeds discrete logic and continuous dynamics into a single augmented system. This transformation enables a unified linear–quadratic (LQ) formulation, bridging the gap between discrete decision-making and continuous regulation. We derive Riccati-type backward recursions that not only characterize optimal continuous control laws but also induce a value-function-driven, state-dependent criterion for logical transitions. Specifically, at each epoch, the optimal logical control is determined by minimizing a quadratic form of the augmented state, where the weighting matrices are recursively updated, providing a principled mechanism for real-time logical-continuous coordination.

\tikzset{
    block/.style = {draw, fill=white, rectangle, minimum height=3em, minimum width=3em},
    tmp/.style  = {coordinate}, 
    sum/.style= {draw, fill=white, circle, node distance=1cm},
    mult/.style= {draw, fill=white, circle, node distance=1cm, minimum size=0.6cm}, % 新增乘法符号样式
    input/.style = {coordinate},
    output/.style = {coordinate},
    pinstyle/.style = {pin edge={to-,thin,black}}
}

\begin{figure}[h!]
\centering
\begin{tikzpicture}[auto, node distance=1.3cm]
    % ==================== 逻辑系统部分 ====================
    \node [input, name=linput] (linput) {};
    \node [input, right of=linput, node distance=0.5cm] (sum) {};
    \node [block, right of=sum] (controller) {$f$};
    \node [input, right of=controller] (lstate) {};
    \node [block, right of=lstate] (holder) {$z^{-1}$};;
    \node [input, right of=holder] (loutput) {};
    \node [input, above of=loutput] (sum1) {};

    % 逻辑系统连接
    \draw [->] (linput) -- node{$\gamma_t$} (controller);
    \draw [->] (controller) --  (holder);
    \draw [-] (holder) -- node{$\theta_t$} (loutput);

    \draw [-] (loutput) -- (sum1);
    \draw [->] (sum1) -| (controller);

    % ==================== 连续系统部分 ====================
    \node [output, below of=lstate] (sum2) {};
    \node [output, below of=sum2, node distance=0.1cm] (subsys1) {};
    \node [output, below of=subsys1] (subsys2) {};
    \node [output, below of=subsys2] (subsysq) {};
    
    % 子系统模块
    \node [block, left of=subsys1] (mode1) {$(A_1,B_1)$};
    \node [block, left of=subsys2] (mode2) {$(A_2,B_2)$};
    \node [block, left of=subsysq] (modeq) {$(A_N,B_N)$};
    
    % 切换和状态更新
    \node [input, right of=subsys1, node distance=0.1cm] (cross2) {};
    \node [input, right of=subsys2, node distance=0.8cm] (switch) {};
    \node [block, right of=switch] (holder1) {$z^{-1}$};
    \node [block, below of=holder] (iota) {$\iota$};
    \node [input, left of=iota, node distance=1cm] (switch2) {};
    \node [input, below of=switch2, node distance=0.5cm] (switch3) {};
    \node [input, right of=holder1] (output) {};
    \node [input, below of=output, node distance=2.1cm] (cross) {};
    \node [input, left of=cross, node distance=6.2cm] (cross1) {};
    
    % 切换控制输入点
    \node [input, above of=switch, node distance=1cm] (switch1) {};

    % ==================== 连续系统连接 ====================
    \draw [->] (switch) --  (holder1);
    \draw [-] (mode1) --  (subsys1);
    \draw [-] (mode2) --  (subsys2);
    \draw [-] (modeq) --  (subsysq);
    \draw [-] (switch) --  (cross2);
    
    % 控制输入
    \node [input, left of=mode1] (sum3) {};
    \node [input, left of=sum3, node distance=0.5cm] (input) {};
    \draw [->] (input) --  node{$u_t$}(mode1);
    \draw [->] (sum3) |-  (mode2);
    \draw [->] (sum3) |-  (modeq);
    
    % 状态输出和反馈
    \draw [-] (holder1) -- node{$x_t$} (output);
    \draw [-] (output) -- (cross);
    \draw [-] (cross) --  (cross1);
    
    % 反馈分支
    \node [input, above of=cross1, node distance=0.5cm] (subq) {};
    \node [input, right of=subq, node distance=0.68cm] (subqq) {};
    \node [input, above of=subq, node distance=1.3cm] (sub2) {};
    \node [input, right of=sub2, node distance=0.77cm] (sub22) {};
    \node [input, above of=sub2, node distance=1.3cm] (sub1) {};
    \node [input, right of=sub1, node distance=0.77cm] (sub11) {};
    
    \draw [-] (cross1) --  (subq); 
    \draw [->] (subq) --  (subqq); 
    \draw [-] (cross1) --  (sub2); 
    \draw [->] (sub2) --  (sub22); 
    \draw [-] (cross1) --  (sub1); 
    \draw [->] (sub1) --  (sub11);
    
    % ==================== 系统间连接 ====================

    \draw [->] (loutput) |- (iota);
    \draw [-] (iota) -- (switch2);
    \draw [->] (switch2) -- (switch3);
    
    % ==================== 连接点（可选，根据需要保留） ====================
    \draw (0,0) circle (1pt);
    \draw (0,-1.4) circle (1pt);
    \draw (3.1,-1.4) circle (1pt);
    \draw (3.1,-2.7) circle (1pt);
    \draw (3.1,-4) circle (1pt);
    
    \draw (3.1,-3.2) circle (0.3pt);
    \draw (3.1,-3.3) circle (0.3pt);
    \draw (3.1,-3.4) circle (0.3pt);
    
\end{tikzpicture}
\caption{{\bf Overall architecture of the proposed co-design framework.} This diagram illustrates the coupling captured algebraically by \eqref{eq5.1}--\eqref{eq5.2} (presented in Section \ref{S2}). 
         The logical control input $\gamma_t$ drives the logical dynamical system $f$, 
         generating the switching signal $\theta_t$ that selects the active continuous 
         subsystem $(A_{\iota(\theta_t)}, B_{\iota(\theta_t)})$. The continuous control input $u_t$ 
         regulates the selected subsystem.\label{fig:system_structure}}
\end{figure}
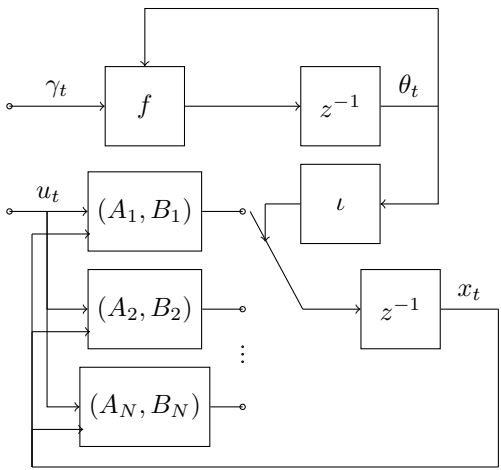

The proposed framework is versatile, accommodating both deterministic and stochastic logical dynamics. We show that classical deterministic free-switching control and the optimal control of Markov Jump Linear Systems (MJLS) are recovered as special cases when logical dynamics are exogenous. When logical control is active, our framework facilitates feedback-driven co-optimization, enabling the logical layer to proactively steer the system toward improved performance.

Main Contributions:
\begin{itemize}
\item {\bf Unified Co-Design via Algebraic Embedding:} We resolve the structural mismatch in hybrid optimization by employing the STP to map logical dynamics into an ASSR. This transformation reformulates the co-design problem as a piecewise quadratic optimization over a unified state space, enabling the derivation of necessary and sufficient optimality conditions in closed form.
\item {\bf Riccati-Based Synthesis with Value-Driven Logic:} For the unified ASSR, we derive backward recursions that yield optimal state-feedback laws for continuous control and state-dependent selection rules for logical control. The resulting cost-to-go matrices serve as quantitative metrics for evaluating discrete actions, effectively coupling forward-looking logical planning with backward dynamic programming.
\item {\bf Hierarchical Algorithm for Real-Time Tractability:} To mitigate the combinatorial explosion of logical decision-making, we propose a hierarchical algorithm that decouples intensive evaluations into an offline precomputation phase. Online execution is thus reduced to an efficient selection from precomputed gains, ensuring computational feasibility for complex system scales.
\end{itemize}

The remainder of this paper is organized as follows. Section~\ref{S2} provides the preliminaries and problem formulation. Section~\ref{S3} presents the main theoretical results and algorithmic development. Section~\ref{S4} validates the framework via numerical experiments, and Section~\ref{S5} concludes the paper.
\section{Problem Formulation\label{S2}}
\emph{Notations}:\begin{itemize}
    \item $\R^n$: the Euclidean space of all real $n$-vectors.
    \item $\mathscr{D}_k:=\{1,2,\cdots,k\}$: the $k$-value logic set.
    \item ${\bf 1}_{n}$: an $n$-dimensional vector with identical entries and $I_{n}$ is an $n$-dimensional identity matrix.
    \item $\delta_n^i$ denotes the $i$-th canonical basis vector of $\mathbb{R}^n$.
    \item $\Delta_n := \{ \delta_n^1, \delta_n^2, \dots, \delta_n^n \}$.
    \item A matrix $L = [\delta_n^{i_1} ~ \delta_n^{i_2} ~ \cdots ~ \delta_n^{i_r}] \in \mathbb{R}^{n \times r}$ is called a logical matrix.
    \item $\mathcal{L}_{s\times r}$: the set of $s\times r$-dimensional logical matrices.
    \item $A\otimes B$ is the Kronecker product of two arbitrary matrices.
    \item $A*B$: the Khatri-Rao product of $A \in{\cal M}_{p\times n}$ and $B\in{\cal M}_{q\times n}$.
    \item $\vec{\theta}_t$, $\vec{\gamma}_t$, $\vec{\sigma}_t$: vector forms of the logical variables $\theta_t$, $\gamma_t$, and $\sigma_t$, respectively.
    \item $\mathbf{u} = \{u_0, \dots, u_{T-1}\}$: sequence of continuous controls.
\item $\bm{\gamma} = \{{\gamma}_0, \dots, {\gamma}_{T-1}\}$: sequence of logical controls.
    \item $\vec{\bm{\gamma}} = \{\vec{\gamma}_0, \dots, \vec{\gamma}_{T-1}\}$: a logical control sequence in vector form.
    \item $\vec{\bm{\gamma}}_{t:T-1} = \{\vec{\gamma}_t, \dots, \vec{\gamma}_{T-1}\}$: subsequence from time $t$ to $T-1$.
\end{itemize}
\subsection{Mathematical Preliminaries}
To enable a unified state-space representation of logical and continuous dynamics, we briefly recall the STP, which provides an algebraic embedding of finite-valued logic into linear spaces.

\begin{dfn}[\cite{che12}] \label{da.1}
Consider two matrices $A\in \R^{m\times n}, B\in \R^{p\times q}$. Let $l$ be the least common multiple of $n$ and $p$. The STP of $A$ and $B$ is defined by
\begin{align*}
A\ltimes B:= \left(A\otimes I_{l/n}\right)\left(B\otimes I_{l/p}\right)\in \R^{ml/n\times ql/p}.
\end{align*}
\end{dfn}

The STP preserves most algebraic properties of the conventional matrix product. In particular, when $n=p$, $A\ltimes B=AB$. Hence, we omit the symbol ``$\ltimes$" when no confusion arises.

The following lemma summarizes those most relevant to our work.

\begin{lem}[\cite{che12}]\label{lem2.0}
\begin{enumerate}
\item \emph{(Swap Rule)} The STP is pseudo-commutative between vectors and matrices. Specifically, if $z\in \mathbb{R}^r$ and $A\in \mathbb{R}^{m\times n}$, then $z \ltimes A = (I_r \otimes A) \ltimes z$.
\item \emph{(Power-Reducing Matrix)} Define $\Phi_k := \diag\left[\delta_k^1, \delta_k^2, \dots, \delta_k^k\right] \in \mathcal{L}_{k^2 \times k}$. For any $x \in \Delta_k$, it holds that $x \ltimes x = \Phi_k x \in \Delta_{k^2}$. The matrix $\Phi_k$ is called the $k$-th power-reducing matrix.
\end{enumerate}
\end{lem}

\subsubsection*{Algebraic Representation of Logical Variables}
To represent logical dynamics algebraically, we identify each logical value $i\in\mathscr{D}_{k}$ with the canonical basis vector $\delta_{k}^i\in\Delta_{k}$. This mapping establishes a one-to-one correspondence between logical values in $\mathscr{D}_k$ and the canonical basis vectors in $\mathbb{R}^k$. For example, for $k=2$ (Boolean case), we have:
\begin{align*}
1 \leftrightarrow \delta_2^1 = \begin{bmatrix}1\\0\end{bmatrix}, \quad
2 \leftrightarrow \delta_2^2 = \begin{bmatrix}0\\1\end{bmatrix}.
\end{align*}

\subsubsection*{Logical Dynamical Systems}
A logical dynamical system is a finite-state machine whose evolution is governed by logical rules. It consists of:
\begin{itemize}
    \item \textbf{State nodes} $\theta^i_t$: discrete variables representing the internal logic state (e.g., mode indicators, switch positions, Boolean flags).
    \item \textbf{Control nodes} $\gamma^j_t$: discrete inputs that can be chosen by the controller to influence the system's evolution.
\end{itemize}
Each node takes values from a finite set $\mathscr{D}_{k_i}$. The system evolves according to logical update functions:
\begin{align}
    \label{eq2.1}
\theta^i_{t+1}=f^i(\theta^1_t,\cdots,\theta^p_t,\gamma^1_t,\cdots,\gamma^q_t),\quad i=1,2,\cdots,p,
\end{align}
where $f^i:\mathscr{D}_{k_1} \times \cdots \times \mathscr{D}_{k_p} \times \mathscr{D}_{h_1} \times \cdots \times \mathscr{D}_{h_q} \to \mathscr{D}_{k_i}$ is the update function of the state note $\theta^i$.

Collectively, \eqref{eq2.1} defines a finite-state controlled dynamical system, where the global logical state evolves causally under the influence of logical control inputs.

\subsubsection*{Global State and Control}
The \textbf{global logical state} is the tuple of all nodal states:
\[
\theta_t:= (\theta_t^1, \theta_t^2, \dots, \theta_t^p) \in \mathscr{D}_{k_1} \times \cdots \times \mathscr{D}_{k_p}.
\]
Similarly, the \textbf{global logical control} is:
\[
\gamma_t:= (\gamma_t^1, \gamma_t^2, \dots, \gamma_t^q) \in \mathscr{D}_{h_1} \times \cdots \times \mathscr{D}_{h_q}.
\]

Using the vector representation, we define:
\begin{align*}
\vec{\theta}_t &:= \ltimes_{i=1}^p \vec{\theta}^i_t \in \Delta_N, \quad N:=\prod_{i=1}^p k_i, \\
\vec{\gamma}_t &:= \ltimes_{j=1}^q \vec{\gamma}^j_t \in \Delta_M, \quad M:=\prod_{j=1}^q h_j.
\end{align*}
Here, $\ltimes_{i=1}^p$ denotes the sequential STP of $p$ vectors. The vector $\vec{\theta}_t$ (resp. $\vec{\gamma}_t$) encodes the joint configuration of all state (resp. control) nodes at time $t$.

\noindent\textbf{Mode index function:} Since $\vec{\theta}_t$ is a canonical basis vector, there exists a unique integer $\iota(\theta_t) \in \{1, 2, \dots, N\}$ such that $\vec{\theta}_t = \delta_N^{\iota(\theta_t)}$. We call $\iota(\theta_t)$ the \textbf{active mode index} (i.e., the index of the active continuous subsystem) at time $t$. This function $\iota(\cdot)$ provides the crucial link to the switched continuous subsystem.

\subsubsection*{Structure Matrix and ASSR}
The key result in the STP framework is that any logical system of the form (\ref{eq2.1}) can be represented as a linear-like system over the vector space of logical states:

\begin{lem}[\cite{che12}]\label{lem2.3}

\begin{itemize}
\item[(i)] Let $f^i:\mathscr{D}_{k_1} \times \cdots \times \mathscr{D}_{k_p} \times \mathscr{D}_{h_1} \times \cdots \times \mathscr{D}_{h_q} \to \mathscr{D}_{k_i}$ be the updating function of $\theta^i$, i.e., $\theta^i_{t+1}=f^i(\theta^1_t,\cdots,\theta^p_t,\gamma^1_t,\cdots,\gamma^q_t)$. Then there exists a unique logical matrix $M_{f^i}\in {\cal L}_{N\times MN}$, called the structure matrix of $f^i$, such that in vector form the logical function can be expressed by
\begin{align}\label{2.1.3}
\vec{\theta}_{t+1}^i=M_{f^i}\ltimes\vec{\gamma}_t\ltimes\vec{\theta}_t.
\end{align}
\item[(ii)] Let $M_{f^i}$ be the structure matrix of the logical function $f^i$, $i=1,2,\cdots,p$. Then there exists a unique logical matrix $L\in {\cal L}_{N\times MN}$ such that in vector form, logical control network (\ref{eq2.1}) can be expressed by
\begin{align}\label{2.1.4.1}
\vec{\theta}_{t+1}=L\ltimes\vec{\gamma}_t\ltimes\vec{\theta}_t,
\end{align}
where $L=M_{f^1}*M_{f^2}*\cdots*M_{f^p}$ is called the structure matrix of the logical control network (\ref{eq2.1}).
\end{itemize}
Equation (\ref{2.1.4.1}) is called the ASSR of the system (\ref{eq2.1}).
\end{lem}

\noindent \textbf{From ASSR to Unified Co-Design:}
The ASSR (\ref{2.1.4.1}) is the cornerstone of our approach. It transforms the \emph{nonlinear} logical dynamics into an algebraic form that is linear in the augmented state for each admissible logical control, thereby enabling a DP formulation in an augmented state space, from which Riccati-type recursions naturally emerge.

\begin{exa}\label{ex:simple_logical}
Consider a Boolean network with two state nodes $\theta^1, \theta^2 \in \mathscr{D}_2$ and one control node $\gamma \in \mathscr{D}_2$:
\begin{align*}
\begin{cases}
\theta^1_{t+1} &= \theta^1_t \wedge \gamma_t, \\
\theta^2_{t+1} &= \theta^2_t \vee \gamma_t,
\end{cases}
\end{align*}
where $\wedge$ and $\vee$ denote logical AND and OR, respectively.

For elementary logic gates (see \cite{che12} for more details):
\[
M_\wedge = \begin{bmatrix}1&0&0&0\\0&1&1&1\end{bmatrix}, \quad
M_\vee = \begin{bmatrix}1&1&1&0\\0&0&0&1\end{bmatrix}.
\]

Using Lemma \ref{lem2.3}, it follows that
\begin{align*}
    \begin{cases}
        \vec{\theta}_{t+1}^1=M_\wedge\left(I_{4}\otimes {\bf 1}_2^\mathrm{T}\right) \vec{\gamma}_t \vec{\theta}_t^1\vec{\theta}_t^2:=M_1\vec{\gamma}_t\vec{\theta}_t,\\
        \vec{\theta}_{t+1}^2=M_\vee\left(I_{2}\otimes {\bf 1}_2^\mathrm{T}\otimes I_{2}\right) \vec{\gamma}_t \vec{\theta}_t^1\vec{\theta}_t^2:=M_2\vec{\gamma}_t\vec{\theta}_t,
    \end{cases}
\end{align*}
where $(I_{4}\otimes {\bf 1}_2^\mathrm{T})$ and $(I_{2}\otimes {\bf 1}_2^\mathrm{T}\otimes I_{2})$ are dummy matrices that remove irrelevant logical variables to match the arity of the logic gate.

Then we have the structure matrix
\[
L = M_1 * M_2 = [\delta_4^1, \delta_4^1, \delta_4^3, \delta_4^3, \delta_4^3, \delta_4^4, \delta_4^3, \delta_4^4 ],
\]
where each column of $L$ corresponds to one of the $2^3=8$ input combinations $(\gamma_t, \theta_t^1, \theta_t^2)$, directly yielding $\vec{\theta}_{t+1} = L \ltimes \vec{\gamma}_t \ltimes \vec{\theta}_t$.

This example illustrates how logical dynamics with control inputs can be systematically embedded into a linear-algebraic form, avoiding ad hoc mode enumeration.
\end{exa}

\begin{rem}
For this simple example, $L$ could also be obtained by enumerating the truth table. The power of the STP framework lies in its \textit{systematic algebraic procedure}, which scales gracefully to large networks and facilitates the analytical control design undertaken in the sequel. \hfill $\Box$
\end{rem}

\subsection{System Model and Control Objective}
We now integrate the logical dynamics with a switched continuous system, resulting in a hybrid system whose switching signal is generated by a controllable logical dynamical system. The core challenge is to jointly design the logical control $\gamma_t$ (which steers the discrete mode transitions) and the continuous control $u_t$ (which regulates the active continuous subsystem) to minimize a quadratic performance index.

Two cases are addressed: deterministic and stochastic logical dynamics.

\subsubsection{Deterministic Case} The system is described by
\begin{align}\label{eq5.1}
\begin{array}{l}
\begin{cases}
\theta_{t+1}=f\big(\theta_t,\gamma_t\big),\\
x_{t+1}=A_{\iota(\theta_t)}x_t+B_{\iota(\theta_t)} u_t,\\
\end{cases}
\end{array}
\end{align}
where $\theta_t \in \mathscr{D}_{k_1} \times \cdots \times \mathscr{D}_{k_p}$ is the global logical state, $\gamma_t\in\mathscr{D}_{h_1} \times \cdots \times \mathscr{D}_{h_q}$ is the global logical control, and $f:\mathscr{D}_{k_1} \times \cdots \times \mathscr{D}_{k_p} \times \mathscr{D}_{h_1} \times \cdots \times \mathscr{D}_{h_q}\rightarrow \mathscr{D}_{k_1} \times \cdots \times \mathscr{D}_{k_p}$ is the global logical update function, which is induced by the nodal rules $(f^1,f^2,\cdots,f^p)$ defined in \eqref{eq2.1}; $x_t\in\mathbb{R}^n$ is the continuous state, $u_t\in\mathbb{R}^m$ is the continuous control, $A_i\in\mathbb{R}^{n\times n}$ and $B_i\in\mathbb{R}^{n\times m}$ define the linear dynamics for mode $i\in\{1,2,\cdots,N\}$. 

The ASSR of the system \eqref{eq5.1} is
\begin{align}\label{eq5.2}
\begin{array}{l}
\begin{cases}
\vec{\theta}_{t+1}=L \ltimes\vec{\gamma}_t \ltimes\vec{\theta}_t,\\
x_{t+1}=\left(\tilde{A}\ltimes\vec{\theta}_t\right)x_t+\left(\tilde{B}\ltimes\vec{\theta}_t\right)u_t,
\end{cases}
\end{array}
\end{align}
where $L\in{\cal L}_{N\times MN}$, $\tilde{A} := [A_1, A_2, \cdots, A_N] \in \mathbb{R}^{n \times Nn}$, $\tilde{B} := [B_1, B_2, \cdots, B_N] \in \mathbb{R}^{n \times Nm}$.

For each fixed logical control $\vec{\gamma}_t$ and realized logical state $\vec{\theta}_t$, the continuous dynamics are linear in $(x_t,u_t)$; the nonlinearity arises only through the discrete evolution of $\vec{\theta}_t$.

\subsubsection{Stochastic Case} We will further study the stochastic system where the state of a randomly switched logical dynamical system steers a switched linear system with additive Gaussian noise:
\begin{align}\label{eq6.1}
\begin{array}{l}
\begin{cases}
\theta_{t+1} = f_{\sigma_t}\big(\theta_t,\gamma_t\big),\\
x_{t+1} = A_{\iota(\theta_t)}x_t+B_{\iota(\theta_t)} u_t+F_{\iota(\theta_t)}w_t.
\end{cases}
\end{array}
\end{align}
where, in addition to the previously defined variables, $\sigma_t \in \{1, \dots, \ell\}$ is an i.i.d. random variable selecting the active logical update rule, distributed according to $\mathbb{P}(\sigma_t = i) = p_i$ with $p_i \in [0,1]$ and $\sum_{i=1}^\ell p_i = 1$. This formulation captures random logic reconfiguration induced by exogenous events such as communication uncertainty, task scheduling, or supervisory decisions. The process noise $w_t \sim \mathcal{N}(0, I_r)$ is independent, and $F_i \in \mathbb{R}^{n \times r}$ are the noise input matrices.

The ASSR of \eqref{eq6.1} is:
\begin{align}\label{eq6.2}
\begin{array}{l}
\begin{cases}
\vec{\theta}_{t+1}=\tilde{L}\ltimes \vec{\sigma}_t\ltimes \vec{\gamma}_t \ltimes \vec{\theta}_t,\\
x_{t+1}=\left(\tilde{A}\ltimes \vec{\theta}_t \right) x_t+\left(\tilde{B}\ltimes \vec{\theta}_t \right) u_t+\left(\tilde{F}\ltimes \vec{\theta}_t\right) w_t,
\end{cases}
\end{array}
\end{align}
where $\tilde{L}:=[L_1,L_2,\cdots,L_\ell]\in{\cal L}_{N\times \ell MN}$, $L_i$ is the structure matrix of $f_i$; $\vec{\sigma}_t\in\Delta_\ell$ is the vector form of the randomly switching signal $\sigma_t$; $\tilde{F} := [F_1, F_2, \cdots, F_N] \in \mathbb{R}^{n \times Nr}$.

\begin{rem}[Controllability Condition]
The results of this paper do not rely on the controllability of individual subsystems $(A_i,B_i)$.
Since the problem concerns finite-horizon performance optimization,
the Riccati recursions derived in the sequel are well-defined under the stated cost assumptions,
even in the presence of uncontrollable modes.
The resulting cost-to-go functions naturally quantify the associated performance limitations.
When logical control is available, the optimal policy may exploit logical decisions
to avoid or shorten visits to such unfavorable modes.
Issues of stabilizability and infinite-horizon optimality are beyond the scope of this paper.
\end{rem}

\subsubsection{Control Objective}
We assume that the hybrid state $(x_t,\theta_t)$ is fully observable at each time $t$.
\textbf{Crucial distinction from free-switching formulations:} Unlike conventional switched control where the switching signal is freely selected at each time step, the logical control $\gamma_t$ in this work acts on a logical dynamical system with intrinsic state evolution, memory, and transition constraints. As a result, switching decisions are not independent hybrid variables but are generated through the evolution of the logical state, introducing a forward-looking coupling between discrete decisions.

The co-design problem is therefore to jointly determine the logical control policy $\gamma_t(\cdot)$ and the continuous control policy $u_t(\cdot)$,
both possibly state-dependent, so as to minimize a quadratic performance index over a finite horizon.
The distinction between open-loop control sequences and state-feedback policies will be made explicit in Section~\ref{S3},
where the dynamic programming formulation is developed on the hybrid state $(x_t,\theta_t)$.

\medskip
\noindent\textbf{Deterministic case.}
We consider the finite-horizon cost functional
\begin{align}\label{cf}
J(x_0,&\theta_0,\mathbf{u},\bm{\gamma}):=
\frac{1}{2}\left[
\sum_{t=0}^{T-1}
\left(
x_t^\mathrm{T}C_{\iota(\theta_t)}x_t
+
u_t^\mathrm{T}D_{\iota(\theta_t)}u_t
\right)
+
x_T^\mathrm{T}Q_{\iota(\theta_T)}x_T
\right],
\end{align}
where $Q_i\succ{\bf 0}$, $C_i\succeq{\bf 0}$, and $D_i\succ{\bf 0}$ are given weighting matrices for all $i\in\{1,\dots,N\}$.

\medskip
\noindent\textbf{Stochastic case.}
When the logical dynamics and continuous subsystem are subject to stochastic disturbances,
we minimize the expected cost functional
\begin{align}\label{cfs}
&\mathcal{J}(x_0,\theta_0,\mathbf{u},\bm{\gamma})
:=\frac{1}{2}\left(
\sum_{t=0}^{T-1}
\mathbb{E}\!\left[
x_t^\mathrm{T}C_{\iota(\theta_t)}x_t
+
u_t^\mathrm{T}D_{\iota(\theta_t)}u_t
\right]
+
\mathbb{E}\!\left[
x_T^\mathrm{T}Q_{\iota(\theta_T)}x_T
\right]
\right),
\end{align}
where the expectation is taken with respect to the joint distribution of the random logical transitions $\{\sigma_t\}$
and the process noise $\{w_t\}$, conditioned on the initial state $(x_0,\theta_0)$ and the applied control policies.

In the next section, we show how the algebraic structure induced by the ASSR
enables a dynamic programming formulation on the augmented state space,
leading to Riccati-type recursions and optimal state-feedback laws for both logical and continuous controls.

\section{Main Results\label{S3}}

\subsection{Co-Design under Deterministic Logical Dynamics}
Recall the deterministic system \eqref{eq5.1} with its ASSR \eqref{eq5.2}. To enable co-design, we define the augmented state $v_t := \vec{\theta}_t \ltimes x_t$. This construction allows the logical state, logical control,
and continuous state to be treated within a unified linear framework. With Lemma \ref{lem2.0}, we can obtain the lifted dynamics:
\begin{align}
    \nonumber
v_{t+1}:=&\vec{\theta}_{t+1}\ltimes x_{t+1}\\\nonumber
=&L\ltimes\vec{\gamma}_t\ltimes\vec{\theta}_t\ltimes\left[ \left(\tilde{A}\ltimes\vec{\theta}_t\right) x_t+\left(\tilde{B}\ltimes\vec{\theta}_t\right) u_t \right]\\\nonumber
=&L\left(I_{MN}\otimes \tilde{A}\right)\left(I_M\otimes\Phi_N\right)\vec{\gamma}_t \vec{\theta}_tx_t+L\left(I_{MN}\otimes\tilde{B}\right)\left(I_M\otimes \Phi_N\right)\vec{\gamma}_t\vec{\theta}_t u_t\\\label{eq5.3} 
:=&{\bf A}\ltimes\vec{\gamma}_t\ltimes v_t+{\bf B}\ltimes \vec{\gamma}_t\ltimes\vec{\theta}_t\ltimes u_t.
\end{align}
The system matrices $\mathbf{A} \in \mathbb{R}^{Nn\times MNn}, \mathbf{B} \in \mathbb{R}^{Nn\times MNm}$ have expanded dimensions to incorporate the effect of the logical control. Under the ASSR, the augmented system is linear in $v_t$ for each admissible logical control $\vec{\gamma}_t$.

With the lifted dynamics \eqref{eq5.3}, the deterministic co-design problem is to jointly optimize the logical and continuous control sequences $\bm{\gamma}$ and $\mathbf{u}$ in order to minimize the quadratic cost $J$ defined in \eqref{cf}. This constitutes a complex combinatorial optimization problem due to the discrete nature of the logical control.

\begin{thm}\label{t5.1}
Consider the system described by the augmented dynamics \eqref{eq5.3}. For the given performance index \eqref{cf}, the following statements hold over the finite horizon $T$:
\begin{enumerate}
    \item An optimal policy pair $(\mathbf{u}^*, \bm{\gamma}^*)$ exists.
    \item For any fixed admissible logical control sequence $\bm{\gamma}$, 
    the corresponding optimal continuous control law $\mathbf{u}^*(\bm{\gamma})$ is uniquely determined.
    It admits a linear state-feedback form whose gain is uniquely determined by a Riccati recursion conditioned on the logical control.
    \item The optimal logical control admits a Bellman-optimal, state-dependent selection rule:
at each time $t$, given the current augmented state $v_t$,
the logical control $\vec{\gamma}_t^*(v_t)$ is selected via the Bellman minimization \eqref{t5.1.3}. This minimization is finite and exact, as $\Delta_M$ is finite.
\end{enumerate}
Moreover, such an optimal policy can be constructed recursively via the following Riccati backward recursion,
which yields a state-feedback representation applicable to all
$v_t\in\Delta_N\times \mathbb{R}^n$.
Initialize with
\[
K^*_T={\bf Q}:=\diag\left(Q_1,Q_2,\cdots,Q_N \right)\in\mathbb{R}^{Nn\times Nn}.
\]
For each $t=T-1,T-2,\cdots,0$, compute
\begin{align}\label{t5.1.1}
 G_t(\vec{\gamma}_t)&=-\left({\bf D}+\vec{\gamma}_t^\mathrm{T}{\bf B}^\mathrm{T}K^*_{t+1}{\bf B}\vec{\gamma}_t \right)^{-1}
 \vec{\gamma}_t^\mathrm{T}{\bf B}^\mathrm{T}K^*_{t+1}{\bf A}\vec{\gamma}_t\in\mathbb{R}^{Nm\times Nn},   
\end{align}
where $G_t(\vec{\gamma}_t)$ acts on the augmented state $v_t=\vec{\theta}_t\ltimes x_t$,
and the actual continuous control input is recovered via $\mathbf{1}_N^\mathrm{T}\ltimes G_t(\vec{\gamma}_t)$.
\begin{align}\label{t5.1.2}
    K_t({\vec{\gamma}_t})=&
    {\bf C}
    +\left({\bf A}\vec{\gamma}_t+{\bf B}\vec{\gamma}_tG_t({\vec{\gamma}_t})\right)^\mathrm{T}
    K_{t+1}^*
    \left({\bf A}\vec{\gamma}_t+{\bf B}\vec{\gamma}_tG_t({\vec{\gamma}_t})\right)+G_{t}({\vec{\gamma}_t})^\mathrm{T}{\bf D}G_{t}({\vec{\gamma}_t})
    \in\mathbb{R}^{Nn\times Nn},
\end{align}
where ${\bf C}:=\diag\left(C_1,\cdots,C_N\right)\in\mathbb{R}^{Nn\times Nn}$
and ${\bf D}:=\diag\left(D_1,\cdots,D_N\right)\in\mathbb{R}^{Nm\times Nm}$.
The optimal logical control at time $t$ is determined by the Bellman minimization\footnote{
For a fixed logical control $\vec{\gamma}_t$, the matrix $K_t(\vec{\gamma}_t)$ represents the cost-to-go
when $\vec{\gamma}_t$ is applied at time $t$ and the optimal policy is followed thereafter.
Its dependence on future logical decisions is captured implicitly through $K_{t+1}^*$.
}
\begin{align}\label{t5.1.3}
\vec{\gamma}_t^*(v_t) := \arg\min_{\vec{\gamma}_t\in\Delta_M}\frac{1}{2}v_t^\mathrm{T}K_t(\vec{\gamma}_t)v_t,
\end{align}
and define $\cal K_t^*$ as the cost-to-go matrix corresponding to the minimizing logical control $\vec{\gamma}_t^*$.
\\
The optimal continuous control is then given by
\begin{align}\label{t5.1.4}
    u^*_t={\bf G}_t({\vec{\gamma}^*_t})v_t,\quad t\in[0,T-1],
\end{align}
where $\mathbf{G}_t(\vec{\gamma}_t) = \mathbf{1}_N^\mathrm{T}\ltimes G_t(\vec{\gamma}_t)\in\mathbb{R}^{m\times Nn}$.
The resulting minimal cost is
\[
J_{min}=\frac{1}{2}v_0^\mathrm{T}K_0^*v_0.
\]
\end{thm}

\emph{Proof.}
See Appendix \ref{app2}.
\hfill $\Box$

The following corollary provides an equivalent interpretation of the Bellman-optimal state-feedback policy in terms of logical control sequences, which is useful for implementation.

\begin{cor}
\label{cor:sequence_opt}
For the deterministic co-design problem, the Bellman-optimal state-feedback policy
characterized in Theorem~\ref{t5.1} admits an equivalent sequence-based interpretation.
Specifically, from any state $v_t$ at time $t=0,1,\cdots,T-1$, the optimal logical control
$\vec{\gamma}_t^*(v_t)$ coincides with the first element of an optimal remaining logical
control sequence $\vec{\bm{\gamma}}_{t:T-1}^*$ that minimizes
\[
\vec{\bm{\gamma}}_{t:T-1}^* = \arg\min_{\vec{\bm{\gamma}}_{t:T-1} \in \mathcal{S}_t}
\frac{1}{2} v_t^\mathrm{T} K_t(\vec{\bm{\gamma}}_{t:T-1}) v_t.
\]
Here, $\mathcal{S}_t$ denotes the set of all admissible logical control sequences from time
$t$ to $T-1$, and $K_t(\vec{\bm{\gamma}}_{t:T-1})$ is obtained by recursively applying
\eqref{t5.1.1}--\eqref{t5.1.2} backward from $K_T=\mathbf{Q}$ using the elements of
$\vec{\bm{\gamma}}_{t:T-1}$.
This equivalence follows directly from Bellman's principle of optimality and does not
constitute an alternative optimization problem.
\end{cor}

\begin{rem}[Relation to Free-Switching Formulations]\label{rem:free_switch}
Theorem~\ref{t5.1} recovers classical optimal control of switched systems with
\emph{free switching} as a special case.
When the logical subsystem has no internal dynamics
($\vec{\theta}_{t+1} = I_N \vec{\gamma}_t$, $M = N$),
the logical state acts as a memoryless mode selector, and the framework reduces
to standard free-switching formulations.

In the general setting considered here, the logical state evolves dynamically.
As a result, switching decisions are evaluated through the Riccati-induced
value function and constrained by the logical state trajectory, thereby
incorporating memory and transition constraints directly into the optimality
conditions, rather than imposing them externally.
\hfill $\Box$
\end{rem}

\begin{rem}[Interpretation and Implementation]\label{rem:interpretation_and_implementation}
Theorem~\ref{t5.1} characterizes a Bellman-optimal state-feedback policy
$\vec{\gamma}_t^*(v_t)$.
Corollary~\ref{cor:sequence_opt} provides an equivalent
\emph{interpretation} of this policy in terms of remaining logical sequences,
which follows directly from Bellman's principle of optimality and does not
constitute an alternative optimization problem.

A key enabler for practical implementation is that, for any fixed admissible
logical subsequence $\vec{\bm{\gamma}}_{t:T-1}$, the matrices
$K_t(\vec{\bm{\gamma}}_{t:T-1})$ and $G_t(\vec{\bm{\gamma}}_{t:T-1})$
depend only on the subsequence and system parameters, and not on the specific
state value $v_t$.
This permits their offline precomputation.
Algorithm~\ref{alg1} exploits this structure by enumerating logical sequences
offline and selecting $\vec{\gamma}_t^*$ online via \eqref{t5.1.3},
yielding a state-feedback policy defined over the entire state space
$\Delta_N \times \mathbb{R}^n$.
\hfill $\Box$
\end{rem}

\begin{algorithm}[h!]
\caption{Precomputation-Based Co-Design Algorithm\label{alg1}}
\hspace*{0.02in} {\bf Input:}
Horizon $T$; system matrices $\mathbf{A}, \mathbf{B}$; weighting matrices $\mathbf{C}, \mathbf{D}, \mathbf{Q}$.\\
\hspace*{0.02in} {\bf Output:}
Optimal logical control sequence $\vec{\bm{\gamma}}^* = (\vec{\gamma}^*_0, \dots, \vec{\gamma}^*_{T-1})$ and associated gain matrices $\{K_t(\cdot), G_t(\cdot)\}$.

\begin{algorithmic}[1]
\STATE Initialize: $\mathcal{S}_T \leftarrow \{\text{empty sequence}\}$
\STATE $K_T^* \leftarrow \mathbf{Q}$
\FOR{$t = T-1$ down to $0$}
    \STATE $\mathcal{S}_t \leftarrow \{\text{empty sequence}\}$
    \FOR{each logical sequence $\vec{\bm{\gamma}}_{t+1:T-1} \in \mathcal{S}_{t+1}$}
        \FOR{each possible $\vec{\gamma}_t \in \Delta_M$}
            \STATE Extend sequence: $\vec{\bm{\gamma}}_{t:T-1} = (\vec{\gamma}_t, \vec{\bm{\gamma}}_{t+1:T-1})$
            \STATE Compute $G_t(\vec{\gamma}_t)$ using \eqref{t5.1.1} with $K_{t+1}^* = K_{t+1}(\vec{\bm{\gamma}}_{t+1:T-1})$
            \COMMENT{Here $G_t(\vec{\gamma}_t)$ depends on $\vec{\gamma}_t$ and the subsequence $\vec{\bm{\gamma}}_{t+1:T-1}$ through $K_{t+1}^*$}
            \STATE Compute $K_t(\vec{\gamma}_t)$ using \eqref{t5.1.2}
            \STATE Store $G_t(\vec{\bm{\gamma}}_{t:T-1}) \leftarrow G_t(\vec{\gamma}_t)$, $K_t(\vec{\bm{\gamma}}_{t:T-1}) \leftarrow K_t(\vec{\gamma}_t)$
            \STATE Add $\vec{\bm{\gamma}}_{t:T-1}$ to $\mathcal{S}_t$
        \ENDFOR
    \ENDFOR
\ENDFOR
\STATE // Online selection for given $v_0$
\STATE Find $\vec{\bm{\gamma}}^* = \arg\min\limits_{\vec{\bm{\gamma}} \in \mathcal{S}_0} \frac{1}{2} v_0^\mathrm{T} K_0(\vec{\bm{\gamma}}) v_0$
\STATE Output $\vec{\bm{\gamma}}^*$ and the associated gain matrices $\{K_t(\vec{\bm{\gamma}}^*_{t:T-1}), G_t(\vec{\bm{\gamma}}^*_{t:T-1})\}_{t=0}^{T-1}$.
\end{algorithmic}
\end{algorithm}

\begin{rem}[Complexity comparison with MIMPC]
    The computational architecture of the proposed hierarchical algorithm fundamentally differs from that of standard Mixed-Integer Model Predictive Control (MIMPC). In MIMPC, the co-design problem is formulated as a mixed-integer quadratic program (MIQP) that must be solved online at each sampling instant \cite{axe07}. The worst-case complexity of such problems grows exponentially with the prediction horizon $T$ and the number of discrete decisions per step, as branch-and-bound methods may need to explore a combinatorial number of nodes before reaching optimality \cite{sho25,ver25}. Indeed, mixed-integer programming is in general NP-hard. By contrast, our approach adopts a fundamentally different architecture: the combinatorial complexity is shifted to an offline precomputation stage, where all possible logical control sequences of length $T$ are enumerated and the corresponding Riccati matrices $\{K_t,G_t\}$ are precomputed. This offline enumeration incurs a complexity of $O(M^T)$, where $M$ is the number of logical control options. However, online execution then reduces to a state-dependent selection from these precomputed solutions, requiring only $O(M)$ evaluations of quadratic forms per time step—independent of the prediction horizon $T$. This eliminates the need for online mixed-integer solvers entirely. The price paid is an offline storage requirement that grows exponentially with $T$, but for moderate horizons typical in receding-horizon implementations, this remains manageable. The resulting online complexity is deterministic and predictable, enabling real-time feasibility even on resource-constrained platforms. This architectural distinction echoes the philosophy of explicit MPC \cite{ken24}, where offline precomputation of the control law enables constant-time online evaluation.
\end{rem}

\subsection{Co-Design under Stochastic Logical Dynamics}
To account for uncertainty and randomness inherent in many practical systems, we extend the co-design framework developed in the deterministic setting
to stochastic logical dynamics.
In this case, the evolution of the logical subsystem is governed by probabilistic transitions,
and the continuous dynamics are affected by additive process noise.

Recall the stochastic switched system~\eqref{eq6.1} and its ASSR~\eqref{eq6.2}, where the logical evolution follows a probabilistic transition mechanism and the continuous dynamics are perturbed by additive noise.

Following the deterministic construction,
we introduce the augmented state $v_t := \vec{\theta}_t \ltimes x_t$.
Then the augmented dynamics can be written as
\begin{align}\nonumber
v_{t+1}=&\tilde{L} \vec{\sigma}_t \vec{\gamma}_t\vec{\theta}_t\tilde{A}\vec{\theta}_tx_t+\tilde{L}\vec{\sigma}_t \vec{\gamma}_t\vec{\theta}_t\tilde{B}\vec{\theta}_t u_t+\tilde{L}  \vec{\sigma}_t \vec{\gamma}_t\vec{\theta}_t\tilde{F}\vec{\theta}_t w_t \\\nonumber
=&\tilde{L}\left(I_{\ell MN}\otimes\tilde{A}\right)\left(I_{\ell M}\otimes \Phi_N\right)\vec{\sigma}_t\vec{\gamma}_t\vec{\theta}_tx_t+\tilde{L}\left(I_{\ell MN}\otimes\tilde{B}\right)\left(I_{\ell M}\otimes \Phi_N\right)\vec{\sigma}_t\vec{\gamma}_t\vec{\theta}_tu_t\\\nonumber
&+\tilde{L}\left(I_{\ell MN}\otimes\tilde{F}\right)\left(I_{\ell M}\otimes \Phi_N\right)\vec{\sigma}_t\vec{\gamma}_t\vec{\theta}_t w_t\\\label{eq6.3}
    :=&\tilde{\bf A}\vec{\sigma}_t\vec{\gamma}_tv_t+\tilde{\bf B}\vec{\sigma}_t\vec{\gamma}_t\vec{\theta}_tu_t+\tilde{\bf F}\vec{\sigma}_t\vec{\gamma}_t\vec{\theta}_t w_t,
\end{align}
where $\tilde{\bf A} \in \mathbb{R}^{Nn\times \ell MNn}$,
$\tilde{\bf B} \in \mathbb{R}^{Nn\times \ell MNm}$,
$\tilde{\bf F} \in \mathbb{R}^{Nn\times \ell MNr}$,
and the symbol ``$\ltimes$'' is omitted for notational simplicity. The augmented dynamics are affine in $(v_t,u_t)$ with multiplicative stochasticity induced by $\sigma_t$.

Within this stochastic framework, the co-design objective is to find the control sequences $\bm{\gamma}$ and $\mathbf{u}$ that minimize the expected cost $\mathcal{J}$ defined in \eqref{cfs}. This requires optimizing the control policy in anticipation of the uncertain system evolution, making it a significantly more challenging problem than its deterministic counterpart.

The structure of the optimal co-design policy under stochastic logical dynamics
is characterized by the following result.
\begin{thm}
    \label{t6.1}
Consider the stochastic system described by the augmented dynamics \eqref{eq6.3}. For the given performance index \eqref{cfs}, 
\begin{enumerate}
\item There exists an optimal policy pair $(\mathbf{u}^*, \bm{\gamma}^*)$
within the class of admissible causal state-feedback policies
that minimizes the expected cost~\eqref{cfs}.
\item For any admissible logical control sequence $\bm{\gamma}$,
the corresponding optimal continuous control law $\mathbf{u}^*(\bm{\gamma})$
is unique and admits a linear state-feedback form.
The feedback gain is uniquely determined by a generalized Riccati-type recursion
that incorporates both stochastic switching and process noise.
\item The optimal logical control admits a Bellman-optimal
state-feedback characterization.
At each time step, for any admissible augmented state $v_t$,
the logical control $\vec{\gamma}_t^*(v_t)$
is selected by minimizing the corresponding expected cost-to-go.
If multiple logical controls attain the same minimum,
any such minimizer is optimal.
\end{enumerate}
Moreover, such an optimal policy can be constructed via a backward generalized Riccati recursion, which is structurally analogous to the deterministic case despite the presence of multiplicative stochasticity and process noise.
Specifically, initialize ${\cal K}^*_T = {\bf Q}$.
For $t = T-1, T-2, \dots, 0$, define
\begin{align}
    {\cal G}_{t}({\vec{\gamma}_{t}})=&-\left[{\bf D}+\vec{\gamma}_{t}^\mathrm{T}\sum\limits_{i=1}^{\ell}p_i\left((\delta_{\ell}^i)^\mathrm{T}\tilde{\bf B}^\mathrm{T}{\cal K}^*_{t+1}\tilde{\bf B}\delta_{\ell}^i\right)\vec{\gamma}_{t}\right]^{-1}\left[\vec{\gamma}_{t}^\mathrm{T}\sum_{i=1}^{\ell}p_i\left((\delta_{\ell}^i)^\mathrm{T}\tilde{\bf B}^\mathrm{T}{\cal K}^*_{t+1}\tilde{\bf A}\delta_{\ell}^i\right)\vec{\gamma}_{t}\right]\in\mathbb{R}^{Nm\times Nn},
\end{align}
\begin{align}\nonumber
    {\cal K}_{t}(\vec{\gamma}_t)={\bf C}+{\cal G}_{t}^\mathrm{T}(\vec{\gamma}_t){\bf D}\mathcal{G}_{t}+\vec{\gamma}^\mathrm{T}_{t}\sum_{i=1}^{\ell}p_i\left((\delta_{\ell}^i)^\mathrm{T}\tilde{\bf A}^\mathrm{T}{\cal K}^*_{t+1}\tilde{\bf A}\delta_{\ell}^i\right)\vec{\gamma}_{t}+\vec{\gamma}^\mathrm{T}_{t}\sum_{i=1}^{\ell}p_i\left((\delta_{\ell}^i)^\mathrm{T}\tilde{\bf A}^\mathrm{T}{\cal K}^*_{t+1}\tilde{\bf B}\delta_{\ell}^i\right)\vec{\gamma}_{t}\mathcal{G}_{t}(\vec{\gamma}_t)\\
    +\mathcal{G}_{t}^\mathrm{T}(\vec{\gamma}_t)\vec{\gamma}_{t}^\mathrm{T}\sum_{i=1}^{\ell}p_i\left((\delta_{\ell}^i)^\mathrm{T}\tilde{\bf B}^\mathrm{T}{\cal K}^*_{t+1}\tilde{\bf A}\delta_{\ell}^i\right)\vec{\gamma}_{t}+\mathcal{G}_{t}^\mathrm{T}(\vec{\gamma}_t)\vec{\gamma}_{t}^\mathrm{T}\sum_{i=1}^{\ell}p_i\left((\delta_{\ell}^i)^\mathrm{T}\tilde{\bf B}^\mathrm{T}{\cal K}^*_{t+1}\tilde{\bf B}\delta_{\ell}^i\right)\vec{\gamma}_{t}\mathcal{G}_{t}(\vec{\gamma}_t)\in\mathbb{R}^{Nn\times Nn}
\end{align}\\
The Bellman value function admits a decomposition into a
state-dependent quadratic form and a state-independent noise-induced constant term, and the optimal logical control is characterized by the Bellman minimization.
\begin{align}\label{t6.1.0}
\vec{\gamma}_t^*(v_t):=\arg\min_{\vec{\gamma}_t\in\Delta_M}
\Bigg[
\frac{1}{2}v_t^\mathrm{T}{\cal K}_t(\vec{\gamma}_t)v_t
+\frac{1}{2}\mathrm{tr}\!\left(
(\vec{\gamma}_{t}\vec{\theta}_{t})^\mathrm{T}
\sum_{i=1}^{\ell}p_i
(\delta_{\ell}^i)^\mathrm{T}\tilde{\bf F}^\mathrm{T}
{\cal K}^*_{t+1}\tilde{\bf F}\delta_{\ell}^i
\,\vec{\gamma}_{t}\vec{\theta}_{t}
\right)
\Bigg],
\end{align}
and define ${\cal K}_t^*$ as the cost-to-go matrix corresponding to the
minimizing logical control $\vec{\gamma}_t^*$.\\The optimal continuous control:
\begin{align}\label{t6.1.1}
    u^*_t=\tilde{\bf {G}}_t(\vec{\gamma}^*_t)v_t,
\end{align}
where $\tilde{\bf {G}}_t(\vec{\gamma}_t) = \mathbf{1}_N^\mathrm{T}\ltimes \mathcal{G}_t(\vec{\gamma}_t)$.\\
The minimal expected cost is\footnote{}
\begin{align}\label{t6.1.2}
\mathcal{J}_{min}=\frac{1}{2}v_0^\mathrm{T}\mathcal{K}^*_0v_0+\frac{1}{2}\sum_{t=0}^{T-1}\mathrm{tr}\left[\left(\vec{\gamma}_{t}^*\vec{\theta}_{t}\right)^\mathrm{T}\sum_{i=1}^{\ell}p_i\left((\delta_{\ell}^i)^\mathrm{T}\tilde{\bf F}^\mathrm{T}\mathcal{K}^*_{t+1}\tilde{\bf F}\delta_{\ell}^i\right)\vec{\gamma}^*_{t}\vec{\theta}_{t}\right].
\end{align}
\end{thm}

\emph{Proof.}
See Appendix \ref{app3}.
\hfill $\Box$

\begin{rem}[Connection to Classical MJLS]
\label{rem:stochastic_mjls}
The stochastic result in Theorem~\ref{t6.1} generalizes the classical optimal control of MJLS as a special case in the absence of logical control. Specifically, the transition probability matrix of a homogeneous Markov chain can be decomposed into a convex combination of finitely many logical matrices, which serve as basis transition structures. By assigning the convex weights as the probability distribution of an exogenous switching signal $\sigma_t$, the logical state $\theta_t$ evolves as an uncontrolled Markov process.

In this special case, the logical evolution reduces to
$\vec{\theta}_{t+1} = \tilde{L}\vec{\sigma}_t \vec{\theta}_t$, with no logical control variable involved. Consequently, the augmented dynamics coincide with those of a standard MJLS under the lifted logical-state representation, and the Riccati-type recursion in Theorem~\ref{t6.1} recovers the classical MJLS optimal control equations. \hfill $\Box$
\end{rem}

\begin{rem}[Implementation via Offline Precomputation and Online Receding-Horizon]\label{rem:stochastic_implementation}
\textbf{Value Function Structure.}
From Theorem~\ref{t6.1}, the Bellman value function at time $t$
can be written as
\[
V_t(v_t)
= \frac{1}{2} v_t^\mathrm{T} \mathcal{K}_t^* v_t + c_t,
\]
where $\mathcal{K}_t^*$ is the optimal Riccati matrix
and $c_t$ collects the cumulative noise-induced cost arising from the
trace terms in \eqref{t6.1.0}.
Importantly, for any given current logical state \(\vec{\theta}_t\) and any fixed future logical control sequence \(\vec{\bm{\gamma}}_{t:T-1}\), the quantity \(c_t(\vec{\gamma}_t, \vec{\theta}_t)\) in \eqref{t6.1.0} represents the \textbf{expected} contribution of future process noise to the cost-to-go. This expectation, which is taken over the random switching signals $\vec{\bm{\sigma}}_{t:T-1}$ and process noise $\bm{w}_{t:T-1}$, is a \textbf{deterministic function} of \(\vec{\theta}_t\) and \(\vec{\bm{\gamma}}_{t:T-1}\) that can be computed offline via the backward recursion in Theorem~\ref{t6.1}.

\textbf{Offline Precomputation.}
For each admissible logical control sequence
$\vec{\bm{\gamma}}_{t:T-1}$, we precompute offline:
\begin{itemize}
    \item the Riccati matrices $\mathcal{K}_t(\vec{\bm{\gamma}}_{t:T-1})$
    and gains $\mathcal{G}_t(\vec{\bm{\gamma}}_{t:T-1})$
    via the backward recursion in Theorem~\ref{t6.1};
    \item the associated noise-induced cost constants
    $c_t(\vec{\bm{\gamma}}_{t:T-1})$, obtained by recursively accumulating
    the trace terms in \eqref{t6.1.0}, with $c_T = 0$.
\end{itemize}

\textbf{Online Receding-Horizon Execution.}
At runtime, given the observed augmented state
$v_t = \vec{\theta}_t \ltimes x_t$, the controller enumerates all
admissible future logical sequences $\vec{\bm{\gamma}}_{t:T-1}$ and
evaluates the total expected cost
\[
\frac{1}{2} v_t^\mathrm{T} \mathcal{K}_t(\vec{\bm{\gamma}}_{t:T-1}) v_t
+ c_t(\vec{\bm{\gamma}}_{t:T-1}),
\]
selecting the sequence that minimizes this quantity.
Only the first logical control $\vec{\gamma}_t^*$ is applied,
yielding a receding-horizon policy that adapts to stochastic realizations.

Algorithm~\ref{alg2} summarizes this online procedure.
\hfill $\Box$
\end{rem}

\begin{algorithm}[h]
\caption{Real-Time Co-Design under Stochastic Logical Dynamics\label{alg2}}
\hspace*{0.02in} {\bf Input:}
Initial augmented state $v_0$; horizon $T$; 
offline-precomputed matrices $\{\mathcal{K}_t(\vec{\bm{\gamma}}_{t:T-1}), \mathcal{G}_t(\vec{\bm{\gamma}}_{t:T-1})\}$ 
for all admissible logical sequences $\vec{\bm{\gamma}}_{0:T-1}$.\\
\hspace*{0.02in} {\bf Output:} Applied logical control sequence $\vec{\bm{\gamma}}^*_{\text{applied}} = (\vec{\gamma}_0^*, \dots, \vec{\gamma}_{T-1}^*)$.

\begin{algorithmic}[1]
\STATE \textbf{Online Receding-Horizon Execution:}
\FOR{$t = 0, 1, \dots, T-1$}
\STATE Observe current state $v_t$
    \STATE $J^* \leftarrow +\infty$
    \FOR{each admissible logical sequence $\vec{\bm{\gamma}}_{t:T-1} = (\vec{\gamma}_t, \dots, \vec{\gamma}_{T-1})$}
        \STATE Retrieve $\mathcal{K}_t(\vec{\bm{\gamma}}_{t:T-1})$ and $\mathcal{G}_t(\vec{\bm{\gamma}}_{t:T-1})$
        \STATE Compute $c_t(\vec{\bm{\gamma}}_{t:T-1})$ using the precomputed
noise-induced cost from the trace term in \eqref{t6.1.0}
        \STATE Compute predicted cost $J = \frac{1}{2} v_t^\mathrm{T} \mathcal{K}_t(\vec{\bm{\gamma}}_{t:T-1}) v_t + c_t$
        \IF{$J < J^*$}
            \STATE $J^* \leftarrow J$
            \STATE $\vec{\gamma}_t^* \leftarrow \vec{\gamma}_t$ \COMMENT{First element of the minimizing sequence}
            \STATE $\mathcal{G}_t^* \leftarrow \mathcal{G}_t(\vec{\bm{\gamma}}_{t:T-1})$ \COMMENT{Store corresponding gain}
        \ENDIF
    \ENDFOR
    \STATE Apply logical control $\vec{\gamma}_t^*$
    \STATE Compute continuous control $u_t^* = \mathbf{1}_N^\mathrm{T} \ltimes \mathcal{G}_t^* v_t$
\ENDFOR
\end{algorithmic}
\end{algorithm}

\section{An Illustrative Example\label{S4}}
This section presents a numerical example designed to highlight the necessity of explicitly modeling logical dynamics in hybrid co-design problems, and to demonstrate how the proposed framework enables logical decisions with memory and dynamic constraints to be optimally integrated into the overall performance index.
\begin{exa}
    \emph{Logical dynamics.}
The Boolean network has two state nodes $\theta_t = (\theta_t^1, \theta_t^2)$, $\theta_t^i \in \mathscr{D}_2$, and one control node $\gamma_t \in \mathscr{D}_2$.

In the deterministic case, the logical update is described by 
\begin{align}\label{exa1}
\begin{cases}
    \theta_{t+1}^1=\gamma_t\odot(\theta_{t}^1\odot \theta_{t}^2),\\
    \theta_{t+1}^2=\gamma_t\odot\theta_{t}^1,
\end{cases}
\end{align}
where $\odot$ denote logical XNOR (with the convention that $1$ represents TRUE and $2$ represents FALSE).

To illustrate how the structure matrix is derived from the logical functions, we enumerate the truth table for \eqref{exa1}. For each combination of $(\gamma_t, \theta_t^1, \theta_t^2)$, we compute the next state $(\theta_{t+1}^1, \theta_{t+1}^2)$ and map it to the corresponding canonical basis vector $\delta_4^i$. For example, when $(\gamma_t, \theta_t^1, \theta_t^2) = (1,1,1)$, we have $\theta_{t+1}^1 = 1$ and $\theta_{t+1}^2 = 1$, so the output is $(1,1) \leftrightarrow \delta_4^1$. Proceeding similarly for all eight input combinations yields the structure matrix
\[
L_1 = \left[\begin{smallmatrix}
  1 & 0 & 0 & 0 & 0 & 0 & 1 & 0 \\
  0 & 0 & 0 & 1 & 0 & 1 & 0 & 0 \\
  0 & 1 & 0 & 0 & 0 & 0 & 0 & 1 \\
  0 & 0 & 1 & 0 & 1 & 0 & 0 & 0
\end{smallmatrix}\right].\]

In the stochastic case, besides the above rule (denoted $f_1$), an additional update rule $f_2$ is introduced:
\begin{align}
\begin{cases}
    \theta_{t+1}^1=\neg\left(\gamma_t\wedge \neg\theta^1_t\wedge \theta^2_t\right),\\
    \theta_{t+1}^2=(\gamma_t\wedge \theta_{t}^1) \oplus \theta_{t}^2,
\end{cases}
\end{align}
where $\wedge$, $\vee$, $\neg$, and $\oplus$ denote logical AND, OR, NOT, and XOR, respectively. 
Following the same truth‑table enumeration (or, equivalently, applying the algebraic procedure of Lemma~\ref{lem2.3}), we obtain its structure matrix
\[
L_2 = \left[\begin{smallmatrix}
  0 & 1 & 0 & 0 & 1 & 0 & 1 & 0 \\
  1 & 0 & 0 & 1 & 0 & 1 & 0 & 1 \\
  0 & 0 & 1 & 0 & 0 & 0 & 0 & 0 \\
  0 & 0 & 0 & 0 & 0 & 0 & 0 & 0
\end{smallmatrix}\right].
\]
The overall structure matrix for the stochastic logical dynamics is therefore $\tilde{L}=\left[L_1,L_2\right]\in{\cal L}_{4\times 32}$. The active update function at each time step is randomly selected from $\{f_1,f_2\}$ by a selector $\sigma_t$ with probabilities $\mathbb{P}(\sigma_t=i)=p_i$, where $(p_1,p_2)=(0.7,0.3)$.

We emphasize that the logical state $\theta_t$ evolves according to its own internal dynamics and cannot be arbitrarily assigned at each time step. This captures a common feature in practical hybrid systems, where mode availability and transitions are constrained by memory, hysteresis, or protocol rules.

\emph{Switched linear dynamics.}
The continuous subsystem is a switched linear system with four modes, where $x_t \in \mathbb{R}^3$, $u_t \in \mathbb{R}^2$, $w_t \in \mathbb{R}^3$. For each mode $i \in [1,4]$, $(A_i, B_i, F_i)$ and $(C_i, D_i, Q_i)$ are specified as follows ($F_i$ is only used in the stochastic setting):
  \begin{align*}
    A_1=\left[\begin{smallmatrix}
        0.8&  0.1& -0.05\\
      0.05& 0.9&  0.1\\
      -0.1& 0.05& 0.85
    \end{smallmatrix}\right],
    B_1=\left[\begin{smallmatrix}
        0.5& 0\\
      0.1& 0.3\\
      0&   0.2
    \end{smallmatrix}\right],
    F_1=\left[\begin{smallmatrix}
        0.1&0&0\\
        0&0.1&0\\
        0&0&0.1
    \end{smallmatrix}\right],\\
C_1=\left[\begin{smallmatrix}
    1.0&  0& 0\\
      0& 0.8& 0\\
      0& 0& 1.2
    \end{smallmatrix}\right],
    D_1=\left[\begin{smallmatrix}
        0.5&0\\
        0&0.5\\
    \end{smallmatrix}\right],
    Q_1=\left[\begin{smallmatrix}
        2.0&0&0\\
        0&1.5&0\\
        0&0&2.5
    \end{smallmatrix}\right].
\end{align*}
    \begin{align*}
    A_2=\left[\begin{smallmatrix}
        0.7&  0.4& 0\\
      -0.4& 0.7&  0.1\\
      0.1& -0.1& 0.8
    \end{smallmatrix}\right],
    B_2=\left[\begin{smallmatrix}
        0.3& 0.2\\
      0.1& 0.4\\
      0.2&   0.1
    \end{smallmatrix}\right],
    F_2=\left[\begin{smallmatrix}
        0.1&0&0\\
        0&0.1&0\\
        0&0&0.1
    \end{smallmatrix}\right],\\
C_2=\left[\begin{smallmatrix}
    1.2&  0& 0\\
      0& 1.0& 0\\
      0& 0& 0.9
    \end{smallmatrix}\right],
    D_2=\left[\begin{smallmatrix}
        0.6&0\\
        0&0.4\\
    \end{smallmatrix}\right],
   Q_2= \left[\begin{smallmatrix}
        1.8&0&0\\
        0&2.0&0\\
        0&0&1.5
    \end{smallmatrix}\right].
\end{align*}
      \begin{align*}
    A_3=\left[\begin{smallmatrix}
        0.9&  0.2& 0.1\\
      0.1& 0.85&  0.15\\
      0.05& 0.1& 0.95
    \end{smallmatrix}\right],
    B_3=\left[\begin{smallmatrix}
        0.8& 0.1\\
      0.2& 0.7\\
      0.1&   0.3
    \end{smallmatrix}\right],
    F_3=\left[\begin{smallmatrix}
        0.1&0&0\\
        0&0.1&0\\
        0&0&0.1
    \end{smallmatrix}\right],\\
C_3=\left[\begin{smallmatrix}
    0.7&  0& 0\\
      0& 1.1& 0\\
      0& 0& 0.8
    \end{smallmatrix}\right],
    D_3=\left[\begin{smallmatrix}
        0.3&0\\
        0&0.7\\
    \end{smallmatrix}\right],
    Q_3=\left[\begin{smallmatrix}
        2.2&0&0\\
        0&1.8&0\\
        0&0&2.0
    \end{smallmatrix}\right].
\end{align*}
        \begin{align*}
    A_4=\left[\begin{smallmatrix}
        0.95&  0.2& -0.1\\
      -0.1& 0.99&  0.2\\
      0.15& -0.1& 0.93
    \end{smallmatrix}\right],
    B_4=\left[\begin{smallmatrix}
        0.4& 0.3\\
      0.2& 0.5\\
      0.3& 0.2
    \end{smallmatrix}\right],
    F_4=\left[\begin{smallmatrix}
        0.1&0&0\\
        0&0.1&0\\
        0&0&0.1
    \end{smallmatrix}\right],\\
C_4=\left[\begin{smallmatrix}
    1.1&  0& 0\\
      0& 0.9& 0\\
      0& 0& 1.0
    \end{smallmatrix}\right],
D_4=\left[\begin{smallmatrix}
        0.4&0\\
        0&0.6\\
    \end{smallmatrix}\right],
    Q_4=\left[\begin{smallmatrix}
        1.5&0&0\\
        0&2.2&0\\
        0&0&1.8
    \end{smallmatrix}\right].
\end{align*}

Using~\eqref{eq5.3}, we construct $\mathbf{A} \in \mathbb{R}^{12 \times 24}$, $\mathbf{B} \in \mathbb{R}^{12 \times 16}$ for the deterministic case. For the stochastic case, $\tilde{\mathbf{A}} \in \mathbb{R}^{12 \times 48}$, $\tilde{\mathbf{B}} \in \mathbb{R}^{12 \times 32}$, and $\tilde{\mathbf{F}} \in \mathbb{R}^{12 \times 48}$ are derived from~\eqref{eq6.3}.

\emph{Implementation.}
All Riccati-type matrices corresponding to admissible logical control sequences are computed offline using Algorithm \ref{alg1}. For $T=3$, this results in $2^t$ candidate matrices at stage $t$, which are stored and reused during online evaluation.

For the deterministic case, given an initial state $v_0$, the optimal logical control sequence is obtained by solving:
\[
\vec{\bm{\gamma}}^*(v_0) = (\vec{\gamma}_0^*, \vec{\gamma}_1^*, \vec{\gamma}_2^*) = \arg\min_{\vec{\bm{\gamma}} \in \mathcal{S}_0} \frac{1}{2} v_0^\mathrm{T} K_0(\vec{\bm{\gamma}}) v_0,
\]
where $\vec{\bm{\gamma}}$ iterates over all admissible logical control sequences of length $3$, and $\mathcal{S}_0$ is the set of all such sequences as constructed in Algorithm~\ref{alg1}.
  
Once the logical control sequence is determined, the corresponding Riccati matrices are uniquely specified, and the optimal control problem is solved.

For the stochastic case, we apply Algorithm~\ref{alg2}. At each time step, based on $v_t$, the controller selects the logical control $\vec{\gamma}_t^*$ and applies the feedback law $u_t^* = \tilde{\mathbf{G}}_t(\vec{\gamma}_t^*) v_t$.

\emph{Results.}
To validate the theoretical results, we present Monte Carlo simulations evaluating the performance of the proposed state-feedback implementation under stochastic dynamics. 
In the deterministic case, Theorem~\ref{t5.1} guarantees that Algorithm~\ref{alg1} attains the exact theoretical optimum \(J = J^*\), and therefore no statistical evaluation is required. 
In contrast, the stochastic setting introduces intrinsic variability due to random logical evolution and process disturbances, which makes a statistical assessment both necessary and informative.

In each Monte Carlo trial, \(x_0\) is uniformly sampled from \([-10,10]^3\), and \(\theta_0\) is drawn uniformly from $\mathscr{D}_4$. For every sample trajectory, we compute the realized cost \(J\) and record the performance ratio \(J/J^*\) against the expected cost \(J^*\).

Fig.~\ref{fig:ratio_mc} shows the empirical distribution of the normalized cost ratio $J/J^*$ over $1000$ Monte Carlo runs with randomly sampled initial conditions. The ratios concentrate tightly around $1$. This concentration indicates that the Riccati-based value function provides a tight characterization of the expected cost, despite stochastic logical transitions and process noise.

\begin{figure}[!h]
    \centering
    \includegraphics[width=.8\textwidth]{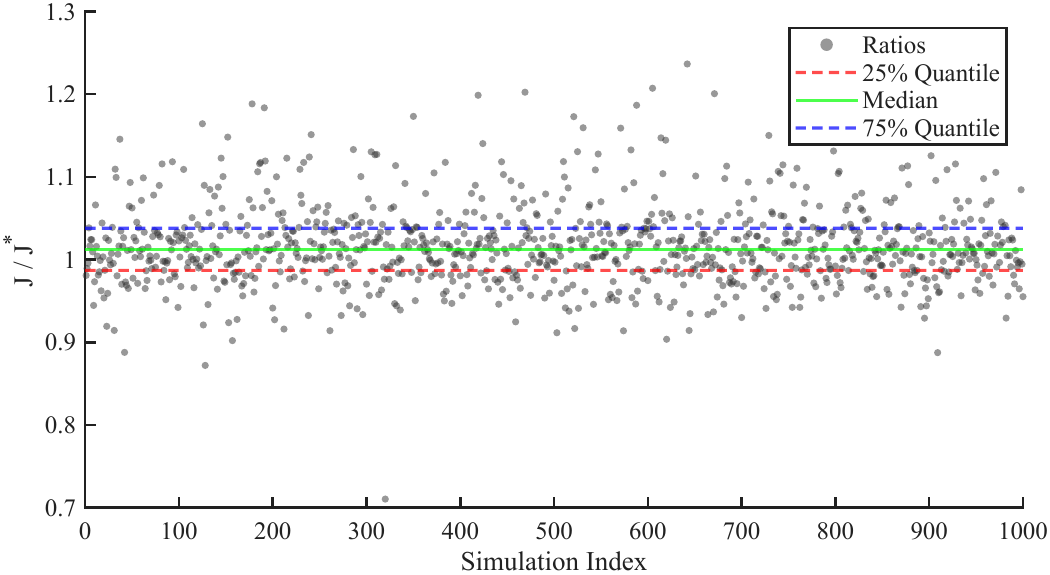}
    \caption{Monte Carlo simulation of the performance ratio $J/J^*$.}
    \label{fig:ratio_mc}
\end{figure}

Fig.~\ref{fig:state_mc} illustrates the evolution of the continuous state trajectories over all $1000$ Monte Carlo trials.  Despite stochastic switching and additive disturbances, all trajectories remain bounded and are driven toward a neighborhood of the origin throughout the simulation horizon.  This empirical behavior indicates robust state regulation under the proposed state-feedback implementation. 
The mean trajectory of each state component is highlighted to illustrate the typical system behavior. This behavior is consistent with the regulation properties implied by the proposed co-design framework.

\begin{figure}[!h]
    \centering
    \includegraphics[width=.8\textwidth]{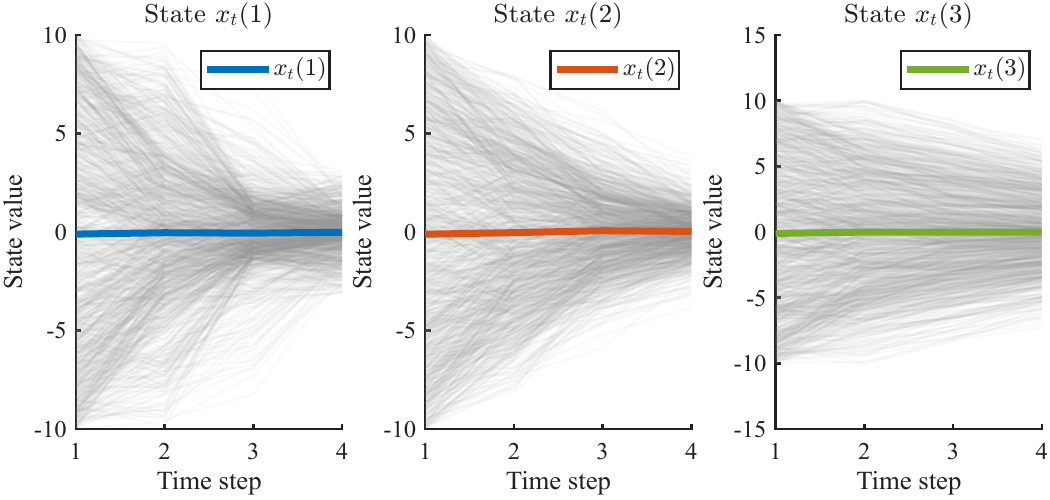}
    \caption{Continuous state trajectories over all Monte Carlo trials.}
    \label{fig:state_mc}
\end{figure}

Notably, such performance guarantees would not be attainable if the logical variables were treated as static mode selectors or instantaneous decision variables, as the implicit memory and transition constraints of the logical layer would be ignored.
\end{exa}

\section{Conclusion\label{S5}}
This paper developed a unified optimal control framework for switched linear systems with logic-driven switching dynamics, enabling the joint co-design of logical and continuous controls. By employing the semi-tensor product to obtain an algebraic state-space representation, the proposed formulation allows classical linear-quadratic and dynamic programming techniques to be systematically extended to this hybrid setting, yielding Riccati-type solutions for both deterministic and stochastic cases.

The framework exhibits notable generality through two distinct degenerations: in the deterministic case, it recovers the classical optimal control of switched systems with freely selectable mode sequences when the logical dynamics reduces to a direct assignment; in the stochastic case, it reduces to the standard optimal control of Markov jump linear systems when the logical control is inactive. These special cases establish a direct connection to the foundational literature on switched and stochastic systems.

\appendix
\section{Auxiliary Lemmas\label{app1}}
Before presenting the proofs, we first establish the STP-based matrix calculus rules needed in the derivations. 

\begin{lem}\label{lem2.1}
Suppose $A \in \mathbb{R}^{p \times q}$ is independent of $x=\left[x_1,x_2,\cdots,x_n\right]^\mathrm{T} \in \mathbb{R}^n$. Let $l$ be the least common multiple of $q$ and $n$. Then
\begin{align}
    \frac{\partial (A \ltimes x)}{\partial x} = A \otimes I_{l/q} \in \mathbb{R}^{pl/q \times l}.
\end{align}
\end{lem}

\emph{Proof.}
    By Definition~\ref{da.1}, we have
\begin{align*}
    A \ltimes x = (A \otimes I_{l/q})(x \otimes I_{l/n})
    := \left[ \hat{A}_1 \; \hat{A}_2 \; \cdots \; \hat{A}_n \right]
    \left[
        \begin{smallmatrix}
            x_1 I_{l/n} \\
            x_2 I_{l/n} \\
            \vdots \\
            x_n I_{l/n}
        \end{smallmatrix}
    \right],
\end{align*}
where $\hat{A}_i = \hat{A} \delta_n^i$, $\hat{A} := A \otimes I_{l/q}$.

Thus, the result can be written as $
    \sum_{i=1}^n x_i \hat{A}_i,$
and it follows that $\hat{A}_i$ is the only term related to $x_i$. Therefore,
\begin{align*}
    \frac{\partial (A \ltimes x)}{\partial x}= \left[ \frac{\partial (A \ltimes x)}{\partial x_1} \; \frac{\partial (A \ltimes x)}{\partial x_2} \; \cdots \; \frac{\partial (A \ltimes x)}{\partial x_n} \right]= \left[ \hat{A}_1 \; \hat{A}_2 \; \cdots \; \hat{A}_n \right] = A \otimes I_{l/q}.
\end{align*}
The proof is thus completed.
\hfill $\Box$

\begin{lem}\label{lem2.2}
Suppose $A \in \mathbb{R}^{p \times q}$ is independent of $x \in \mathbb{R}^n$, $l = \mathrm{lcm}(q, n)$, and $t = \mathrm{lcm}(pl/q, n)$. Then
\begin{align}\nonumber
    \frac{\partial (x^\mathrm{T} \ltimes A \ltimes x)}{\partial x}
    =\left[ (\delta_n^1)^\mathrm{T} A x  (\delta_n^2)^\mathrm{T} A x \cdots (\delta_n^n)^\mathrm{T} A x \right]+ x^\mathrm{T} \hat{A} \in \mathbb{R}^{t/n \times tq/p},
\end{align}
where $\hat{A} = A \otimes I_{l/q}$.
\end{lem}

\emph{Proof.}
\begin{align*}
    x^\mathrm{T} \ltimes A \ltimes x
    &= \left( \sum_{i=1}^n x_i (\delta_n^i)^\mathrm{T} \otimes I_{t/n} \right) \left[ \left( \sum_{i=1}^n x_i \hat{A}_i \right) \otimes I_{tq/pl} \right],
\end{align*}
where $\hat{A}_i = \hat{A} \delta_n^i$ with $\hat{A}_i \in \mathbb{R}^{pl/q \times l/n}$, $i \in [1, n]$.
Taking the derivative with respect to $x_i$ yields
\begin{align*}
\left[ (\delta_n^i)^\mathrm{T} \otimes I_{t/n} \right] \left[ \left( \sum_{j=1}^n x_j \hat{A}_j \right) \otimes I_{tq/pl} \right]+ \left[ \sum_{j=1}^n x_j (\delta_n^j)^\mathrm{T} \otimes I_{t/n} \right] \left( \hat{A}_i \otimes I_{tq/pl} \right)= (\delta_n^i)^\mathrm{T} A x + x^\mathrm{T} \hat{A}_i.
\end{align*}
The result can be then obtained.
\hfill $\Box$

\section{Proof of Theorem \ref{t5.1}\label{app2}}
\textbf{Proof Structure Overview:} The proof follows the standard DP principle, executed backward in time over the entire augmented state space $\Delta_N \times \mathbb{R}^n$. At each stage $t$, the optimization proceeds in two nested steps for \textbf{every possible state} $v_t \in \Delta_N \times \mathbb{R}^n$:

\begin{enumerate}
    \item \textbf{Inner minimization (continuous control):} For each admissible logical control $\vec{\gamma}_t \in \Delta_M$, we first minimize the cost-to-go with respect to the continuous control $u_t$, obtaining a conditional optimal continuous control law $u_t^*(\vec{\gamma}_t, v_t)$.
    \item \textbf{Outer minimization (logical control):} We then minimize the resulting conditional value function with respect to $\vec{\gamma}_t$, yielding the optimal logical control $\vec{\gamma}_t^*(v_t)$.
\end{enumerate}

This two-step minimization is performed for all $v_t \in \Delta_N \times \mathbb{R}^n$, ensuring that the resulting policy is defined over the entire state space, not merely along a specific trajectory. The key insight enabling tractability is that, for any fixed $\vec{\gamma}_t$, the optimal continuous control gain $G_t(\vec{\gamma}_t)$ and the value function matrix $K_t(\vec{\gamma}_t)$ are computed independently of the specific state value $v_t$.

\emph{Proof.} We first establish the existence and uniqueness properties.

\textbf{Existence:} The optimization problem consists of minimizing the quadratic cost \eqref{cf} subject to the augmented dynamics \eqref{eq5.3} over a finite horizon $T$. 
The decision variables are the logical control sequence $\vec{\bm{\gamma}} = \{\vec{\gamma}_0, \dots, \vec{\gamma}_{T-1}\}$ with each $\vec{\gamma}_t \in \Delta_M$ (a finite set) and the continuous control sequence $\bm{u} = \{u_0, \dots, u_{T-1}\}$ with each $u_t \in \mathbb{R}^m$. 
For any fixed logical sequence $\vec{\gamma}$, the continuous subproblem is a standard linear-quadratic regulator (LQR) problem with time-varying matrices, which admits a unique optimal solution when $D_i \succ 0$ and $Q_i \succ 0$. 
Since there are only finitely many possible logical sequences ($M^T$), the global minimum over both $\vec{\bm{\gamma}}$ and $\bm{u}$ is attained by selecting the best among these finite LQR solutions.

\textbf{Uniqueness:}
\begin{enumerate}
    \item For a fixed logical sequence $\vec{\bm{\gamma}}$, the continuous control law $u_t^*(\vec{\gamma})$ is uniquely determined by the Riccati recursion, as the LQR problem has a strictly convex cost in $\mathbf{u}$.
    \item The optimal logical sequence may not be unique; multiple sequences can yield the same minimal cost. However, the \emph{value function} (the minimal achievable cost from any state) is unique. The feedback law $\vec{\gamma}_t^*(v_t)$ defined in the theorem selects \emph{any} minimizer at each state, and all such choices yield the same optimal cost.
\end{enumerate}

We now derive the recursive equations that characterize the optimal policy. Define the value function at terminal time $T$ as
\begin{align*}
  V_T(x_T,\vec{\theta}_T):=\frac{1}{2}x^\mathrm{T}_TQ_{\iota(\theta_T)}x_T=\frac{1}{2}v^\mathrm{T}_T{\bf Q}v_T.
\end{align*}
Here, the value function $V_t$ is defined for all possible augmented states $v_t$ in the augmented state space $\Delta_N\times \mathbb{R}^n$. The dynamic programming recursion will derive a feedback law that is valid for any such state.

Now we analyze the value function of time $T-1$:
\begin{align*}
    V_{T-1}=\frac{1}{2}&{\Big[}v_{T-1}^\mathrm{T}{\bf C} v_{T-1}+u_{T-1}^\mathrm{T}\left(\vec{\theta}_{T-1}^\mathrm{T}{\bf D}\vec{\theta}_{T-1}\right) u_{T-1}\\
    &+\left({\bf A}\vec{\gamma}_{T-1}v_{T-1}+{\bf B}\vec{\gamma}_{T-1}\vec{\theta}_{T-1}u_{T-1}\right)^\mathrm{T}{\bf Q}\left({\bf A}\vec{\gamma}_{T-1}v_{T-1}+{\bf B}\vec{\gamma}_{T-1}\vec{\theta}_{T-1}u_{T-1}\right){\Big]}.
\end{align*}
For any fixed logical control $\vec{\gamma}_{T-1}$, the minimization over the continuous input $u_{T-1}$ is a quadratic problem. Thus, for a given $\vec{\gamma}_{T-1}$, we take the derivative of the value function with respect to $u_{T-1}$ (using Lemmas \ref{lem2.1} and \ref{lem2.2}):
\begin{align*}
    \frac{\partial V_{T-1}}{\partial u_{T-1}}&=u_{T-1}^\mathrm{T}\left[\vec{\theta}_{T-1}^\mathrm{T}\left({\bf D}+\vec{\gamma}_{T-1}^\mathrm{T}{\bf B}^\mathrm{T}{\bf Q}{\bf B}\vec{\gamma}_{T-1}\right)\vec{\theta}_{T-1}\right]+v^\mathrm{T}_{T-1}\left(\vec{\gamma}_{T-1}^\mathrm{T}{\bf A}^\mathrm{T}{\bf Q}{\bf B}\vec{\gamma}_{T-1}\right)\vec{\theta}_{T-1}\in\mathbb{R}^{1\times m},\\
    \frac{\partial^2 V_{T-1}}{\partial u^2_{T-1}}&=\vec{\theta}_{T-1}^\mathrm{T}\left({\bf D}+\vec{\gamma}_{T-1}^\mathrm{T}{\bf B}^\mathrm{T}{\bf Q}{\bf B}\vec{\gamma}_{T-1} \right)\vec{\theta}_{T-1}\succ {\bf 0}_{m\times m},
\end{align*}
confirming convexity in $u_{T-1}$.
Setting this gradient to zero yields the optimal continuous control law conditional on $\vec{\gamma}_{T-1}$:
\begin{align*}
    \vec{\theta}_{T-1}u^*_{T-1}({\vec{\gamma}_{T-1}})=-\left({\bf D}+\vec{\gamma}_{T-1}^\mathrm{T}{\bf B}^\mathrm{T}{\bf Q}{\bf B}\vec{\gamma}_{T-1} \right)^{-1}\left(\vec{\gamma}_{T-1}^\mathrm{T}{\bf B}^\mathrm{T}{\bf Q}{\bf A}\vec{\gamma}_{T-1}\right)v_{T-1}
    :=G_{T-1}({\vec{\gamma}_{T-1}})v_{T-1}.
    \end{align*} Note that $\vec{\theta}_{T-1}u^*_{T-1}({\vec{\gamma}_{T-1}})$ represents a stacked vector containing $u^*_{T-1}({\vec{\gamma}_{T-1}})$ for each possible mode. Multiplying by $\mathbf{1}_N^\mathrm{T}$ extracts $u^*_{T-1}({\vec{\gamma}_{T-1}})$ from the stacked vector, since $\vec{\theta}_{T-1}$ is a canonical basis vector.
\begin{align*}
    u^*_{T-1}({\vec{\gamma}_{T-1}})=\mathbf{1}_N^\mathrm{T}\ltimes G_{T-1}({\vec{\gamma}_{T-1}})v_{T-1}:={\bf G}_{T-1}({\vec{\gamma}_{T-1}})v_{T-1}.
\end{align*}
Substituting this conditional control law back yields a value function parameterized by $\vec{\gamma}_{T-1}$:
\begin{align*}
    V^*_{T-1}(v_{T-1},\vec{\gamma}_{T-1})=&\frac{1}{2}v_{T-1}^\mathrm{T}\Bigg[{\bf C}+G_{T-1}^\mathrm{T}({\vec{\gamma}_{T-1}}){\bf D}G_{T-1}({\vec{\gamma}_{T-1}})+\left({\bf A}+{\bf B}\vec{\gamma}_{T-1}G_{T-1}({\vec{\gamma}_{T-1}})\right)^\mathrm{T}{\bf Q}\\&\times\left({\bf A}+{\bf B}\vec{\gamma}_{T-1}G_{T-1}({\vec{\gamma}_{T-1}})\right)\Bigg]v_{T-1}\\
    :=&\frac{1}{2}v_{T-1}^\mathrm{T}K_{T-1}({\vec{\gamma}_{T-1}})v_{T-1}.
\end{align*}
This results in a cost-to-go function parameterized by $\vec{\gamma}_{T-1}$. It is crucial to note that the matrix 
 $K_{T-1}({\vec{\gamma}_{T-1}})$ is independent of the specific value of $v_{T-1}$. It is a function only of the logical control $\vec{\gamma}_{T-1}$ and the system matrices. Therefore, the above expression represents the cost-to-go from any given state $v_{T-1}$ if logical control $\vec{\gamma}_{T-1}$ is applied at this step.

According to the principle of dynamic programming, the optimal logical control at time $T-1$ should be chosen as a feedback law that maps any state $v_{T-1}$ to a logical decision. Since $K_{T-1}(\vec{\gamma})$ is a symmetric positive definite matrix for each $\vec{\gamma}\in\Delta_M$, the optimal choice for a given $v_{T-1}$ is simply the $\vec{\gamma}$ that minimizes the above quadratic form:
\begin{align*}
    \vec{\gamma}_{T-1}^*:=\arg\min\limits_{\vec{\gamma_{T-1}}\in\Delta_{M}} \frac{1}{2}v_{T-1}^\mathrm{T}K_{T-1}(\vec{\gamma}_{T-1})v_{T-1},
\end{align*}
and denote $K_{T-1}^*:=K_{T-1}({\vec{\gamma}^*_{T-1}})$.

Iterating this recursion backward provides the general solution for all $t = T-1, T-2, \dots, 0$.

The above recursion yields a set of matrices $\{K^*_t\}$ and control gains $\{G^*_t\}$, which together define the optimal feedback policy for all possible states $v_t$, as stated in Theorem \ref{t5.1}.
\hfill $\Box$

\section{Proof of Theorem \ref{t6.1}\label{app3}}
\emph{Proof.} \textbf{Existence and Uniqueness (Stochastic Case):} 
The stochastic co-design problem minimizes the expected quadratic cost \eqref{cfs} subject to \eqref{eq6.3}. 
For any fixed logical control sequence $\vec{\bm{\gamma}}$, the continuous subproblem becomes a finite-horizon linear-quadratic-Gaussian (LQG) problem with Markovian switching. 
Under the conditions $D_i \succ 0$ and $Q_i \succ 0$, this LQG problem is strictly convex in $\bm{u}$, guaranteeing a unique optimal continuous control law (in the sense of almost-sure equivalence) given by the generalized Riccati recursion. 
Since the set of logical sequences is finite, an optimal pair $(\vec{\bm{\gamma}}^*, \bm{u}^*)$ exists by comparing the finite number of expected costs from the corresponding LQG solutions.

The uniqueness considerations are analogous to the deterministic case: the continuous law is unique for a given $\vec{\gamma}\in\Delta_M$, while the logical law $\vec{\gamma}_t^*(v_t)$ selects any minimizer of the expected cost-to-go, with the minimal expected value function being unique.

The dynamic programming recursion for the stochastic case proceeds as follows.

At time $T$, ${\cal V}_T=\frac{1}{2}v_T^\mathrm{T}{\bf Q}v_T$, then given $v_{T-1}$,
\begin{align}\label{eq.app0}
{\cal V}_{T-1}=&\frac{1}{2}\left[v^\mathrm{T}_{T-1}{\bf C}v_{T-1}+u_{T-1}^\mathrm{T}\left(\vec{\theta}_{T-1}^\mathrm{T}{\bf D}\vec{\theta}_{T-1}\right)u_{T-1}\right]+{\cal V}_T.
\end{align}
Given the i.i.d. switching probabilities $\mathbb{P}(\sigma_t = i) = p_i$ and the noise $w_t \sim \mathcal{N}(0, I_r)$, we obtain from \eqref{eq6.3} that
\begin{align}\nonumber
\mathbb{E}[{\cal V}_{T}\mid v_{T-1}]
  =&\frac{1}{2}\Bigg\{\left(\vec{\gamma}_{T-1}^\mathrm{T}v_{T-1}^\mathrm{T}\right)\sum_{i=1}^{\ell}p_i\left((\delta_{\ell}^i)^\mathrm{T}\tilde{\bf A}^\mathrm{T}{\bf Q}\tilde{\bf A}\delta_{\ell}^i\right)\vec{\gamma}_{T-1}v_{T-1}\\\nonumber
&+\left(\vec{\gamma}_{T-1}^\mathrm{T}\vec{\theta}_{T-1}^\mathrm{T}u_{T-1}^\mathrm{T}\right)\sum_{i=1}^{\ell}p_i\left((\delta_{\ell}^i)^\mathrm{T}\tilde{\bf B}^\mathrm{T}{\bf Q}\tilde{\bf B}\delta_{\ell}^i\right)\vec{\gamma}_{T-1}\vec{\theta}_{T-1}u_{T-1}\\\nonumber
  &+\left(\vec{\gamma}_{T-1}^\mathrm{T}v_{T-1}^\mathrm{T}\right)\sum_{i=1}^{\ell}p_i\left((\delta_{\ell}^i)^\mathrm{T}\tilde{\bf A}^\mathrm{T}{\bf Q}\tilde{\bf B}\delta_{\ell}^i\right)\vec{\gamma}_{T-1}\vec{\theta}_{T-1}u_{T-1}\\\nonumber
  &+\left(\vec{\gamma}_{T-1}^\mathrm{T}\vec{\theta}_{T-1}^\mathrm{T}u_{T-1}^\mathrm{T}\right)\sum_{i=1}^{\ell}p_i\left((\delta_{\ell}^i)^\mathrm{T}\tilde{\bf B}^\mathrm{T}{\bf Q}\tilde{\bf A}\delta_{\ell}^i\right)\vec{\gamma}_{T-1}v_{T-1}
  \\\label{eq.app1}
  &+\mathrm{tr}\left[\left(\vec{\gamma}^\mathrm{T}_{T-1}\vec{\theta}^\mathrm{T}_{T-1}\right)\sum_{i=1}^{\ell}p_i\left((\delta_{\ell}^i)^\mathrm{T}\tilde{\bf F}^\mathrm{T}{\bf Q}\tilde{\bf F}\delta_{\ell}^i\right)\vec{\gamma}_{T-1}\vec{\theta}_{T-1}\right]\Bigg\}.
\end{align}

Then given $\vec{\gamma}_{T-1}$, we have
\begin{align*}
    \frac{\partial {\cal V}_{T-1}}{\partial u_{T-1}}&=v^\mathrm{T}_{T-1}\left(\vec{\gamma}_{T-1}^\mathrm{T}\sum_{i=1}^{\ell}p_i\left((\delta_{\ell}^i)^\mathrm{T}\tilde{\bf A}^\mathrm{T}{\bf Q}\tilde{\bf B}\delta_{\ell}^i\right)\vec{\gamma}_{T-1}\right)\vec{\theta}_{T-1}\\
    &+u_{T-1}^\mathrm{T}\left[\vec{\theta}_{T-1}^\mathrm{T}\left({\bf D}+\vec{\gamma}_{T-1}^\mathrm{T}\sum_{i=1}^{\ell}p_i\left((\delta_{\ell}^i)^\mathrm{T}\tilde{\bf B}^\mathrm{T}{\bf Q}\tilde{\bf B}\delta_{\ell}^i\right)\vec{\gamma}_{T-1}\right)\vec{\theta}_{T-1}\right]
    \end{align*}
    and the Hessian $\frac{\partial^2 {\cal V}_{T-1}}{\partial u^2_{T-1}}\in \mathbb{R}^{m\times m}$ is positive-definite:

    \begin{align*}
\vec{\theta}_{T-1}^\mathrm{T}\left({\bf D}+\vec{\gamma}_{T-1}^\mathrm{T}\sum_{i=1}^{\ell}p_i\left((\delta_{\ell}^i)^\mathrm{T}\tilde{\bf A}^\mathrm{T}{\bf Q}\tilde{\bf B}\delta_{\ell}^i\right)\vec{\gamma}_{T-1} \right)\vec{\theta}_{T-1}\succ {\bf 0}.
\end{align*}
It follows that $
\vec{\theta}_{T-1}u^*_{T-1}(\vec{\gamma}_{T-1})={\cal G}_{T-1}({\vec{\gamma}_{T-1}})v_{T-1}$, and
$u^*_{T-1}(\vec{\gamma}_{T-1})=\tilde{\bf G}_{T-1}({\vec{\gamma}_{T-1}})v_{T-1}$, where 
\begin{align*}
    {\cal G}_{T-1}({\vec{\gamma}_{T-1}}):=&-\left[{\bf D}+\vec{\gamma}_{T-1}^\mathrm{T}\sum\limits_{i=1}^{\ell}p_i\left((\delta_{\ell}^i)^\mathrm{T}\tilde{\bf B}^\mathrm{T}{\bf Q}\tilde{\bf B}\delta_{\ell}^i\right)\vec{\gamma}_{T-1}\right]^{-1}\left[\vec{\gamma}_{T-1}^\mathrm{T}\sum_{i=1}^{\ell}p_i\left((\delta_{\ell}^i)^\mathrm{T}\tilde{\bf B}^\mathrm{T}{\bf Q}\tilde{\bf A}\delta_{\ell}^i\right)\vec{\gamma}_{T-1}\right],\\
    \tilde{\bf G}_{T-1}({\vec{\gamma}_{T-1}})=&{\bf 1}_N^\mathrm{T}\ltimes {\cal G}_{T-1}({\vec{\gamma}_{T-1}}).
\end{align*}
Analogous to the proof of Theorem \ref{t5.1}, the row vector ${\bf 1}_N^\mathrm{T}$ extracts the optimal continuous control $u^*_{T-1}$ from the stacked vector $\vec{\theta}_{T-1}u^*_{T-1}$. This is valid because, despite the stochastic switching, the current logical state $\vec{\theta}_{T-1}$ remains a canonical basis vector $\delta_N^i$, and the multiplication by ${\bf 1}_N^\mathrm{T}$ effectively selects the control gain block corresponding to the active mode $i$.

Substituting $u_{T-1}^*({\vec{\gamma}_{T-1}})$ and \eqref{eq.app1} back into \eqref{eq.app0} yields
\begin{align*}
    \mathbb{E}[{\cal V}^*_{T-1}]=\frac{1}{2}v_{T-1}^\mathrm{T}{\cal K}_{T-1}({\vec{\gamma}_{T-1}})v_{T-1}+\frac{1}{2}\mathrm{tr}\left[\left(\vec{\gamma}_{T-1}\vec{\theta}_{T-1}\right)^\mathrm{T}\sum_{i=1}^{\ell}p_i\left((\delta_{\ell}^i)^\mathrm{T}\tilde{\bf F}^\mathrm{T}{\bf Q}\tilde{\bf F}\delta_{\ell}^i\right)\vec{\gamma}_{T-1}\vec{\theta}_{T-1}\right],
\end{align*}
where (for notational simplicity, we omit the dependency on $\vec{\gamma}_t$ in the following)
\begin{align*}
    {\cal K}_{T-1}=&{\bf C}+{\cal G}_{T-1}^\mathrm{T}{\bf D}\mathcal{G}_{T-1}+\vec{\gamma}^\mathrm{T}_{T-1}\sum_{i=1}^{\ell}p_i\left((\delta_{\ell}^i)^\mathrm{T}\tilde{\bf A}^\mathrm{T}{\bf Q}\tilde{\bf A}\delta_{\ell}^i\right)\vec{\gamma}_{T-1}    +\vec{\gamma}^\mathrm{T}_{T-1}\sum_{i=1}^{\ell}p_i\left((\delta_{\ell}^i)^\mathrm{T}\tilde{\bf A}^\mathrm{T}{\bf Q}\tilde{\bf B}\delta_{\ell}^i\right)\vec{\gamma}_{T-1}\mathcal{G}_{T-1}\\
    &+\mathcal{G}_{T-1}^\mathrm{T}\vec{\gamma}_{T-1}^\mathrm{T}\sum_{i=1}^{\ell}p_i\left((\delta_{\ell}^i)^\mathrm{T}\tilde{\bf B}^\mathrm{T}{\bf Q}\tilde{\bf A}\delta_{\ell}^i\right)\vec{\gamma}_{T-1}+\mathcal{G}_{T-1}^\mathrm{T}\vec{\gamma}_{T-1}^\mathrm{T}\sum_{i=1}^{\ell}p_i\left((\delta_{\ell}^i)^\mathrm{T}\tilde{\bf B}^\mathrm{T}{\bf Q}\tilde{\bf B}\delta_{\ell}^i\right)\vec{\gamma}_{T-1}\mathcal{G}_{T-1}.
\end{align*}
To minimize $\mathbb{E}[{\cal V}^*_{T-1}]$ given $v_{T-1}$, let $\vec{\gamma}_{T-1}=\vec{\gamma}_{T-1}^*$ where
\begin{align*}
    \vec{\gamma}_{T-1}^*:=\arg\min\limits_{\vec{\gamma}_{T-1}\in\Delta_{M}} \frac{1}{2}v_{T-1}^\mathrm{T}{\cal K}_{T-1}({\vec{\gamma}_{T-1}})v_{T-1}+\frac{1}{2}\mathrm{tr}\left[\left(\vec{\gamma}_{T-1}\vec{\theta}_{T-1}\right)^\mathrm{T}\sum_{i=1}^{\ell}p_i\left((\delta_{\ell}^i)^\mathrm{T}\tilde{\bf F}^\mathrm{T}{\bf Q}\tilde{\bf F}\delta_{\ell}^i\right)\vec{\gamma}_{T-1}\vec{\theta}_{T-1}\right].
\end{align*}

The optimal policies and value functions for all preceding time steps $t = T-2,\dots, 0$ are obtained analogously by backward induction. This completes the proof.

\noindent\emph{Remark on the role of the trace term.} 
Note that for a fixed logical control $\vec{\gamma}_t$, the trace term is independent of the continuous control $u_t$. 
Consequently, when minimizing $\mathbb{E}[\mathcal{V}_{t+1} \mid v_t, \vec{\gamma}_t, u_t]$ over $u_t$, this term does not appear in the gradient or Hessian with respect to $u_t$. 
Therefore, the optimal continuous control gain $\mathcal{G}_t(\vec{\gamma}_t)$ and the Riccati matrix $\mathcal{K}_t(\vec{\gamma}_t)$ are determined solely by the deterministic part of the dynamics (the $\tilde{\mathbf{A}}$ and $\tilde{\mathbf{B}}$ terms), unaffected by the noise statistics. 

However, the trace term does depend on $\vec{\gamma}_t$ and $\vec{\theta}_t$. 
Thus, when selecting the optimal logical control $\vec{\gamma}_t^*$, this term must be included in the minimization, as it contributes to the expected cost-to-go. 
This separation is reflected in the two components of the optimal logical control criterion \eqref{t6.1.1}: 
the quadratic form accounts for the expected evolution, while the trace term accounts for the noise effect.
\hfill $\Box$

\end{document}